\documentclass[a4paper]{llncs}
\usepackage{fullpage}

\usepackage{amsmath}
\usepackage{amssymb}
\usepackage{graphicx}
\usepackage{xspace}
\usepackage{wrapfig}
\usepackage{enumerate}
\usepackage{comment}
\usepackage{cite}
\usepackage{algorithm}
\usepackage{subfig}
\usepackage{array}
\usepackage[noend]{algpseudocode}
\usepackage[usenames]{xcolor}
\usepackage{stmaryrd} 
\usepackage{mathtools} 
\usepackage{todonotes}

\spnewtheorem{prop}{Property}{\bfseries}{\itshape} 

\newcommand{\remove}[1]{}

\renewenvironment{proof}
{{\em Proof.\ }}{\hspace*{\fill}$\Box$\par\vspace{2mm}}

\renewcommand{\epsilon}{\varepsilon}

\begin{document}
\title{LR-Drawings of Ordered Rooted Binary Trees and\\ Near-Linear Area Drawings of Outerplanar Graphs 
\thanks{Research partially supported by MIUR Project MODE.}
}
\author{Fabrizio Frati, Maurizio Patrignani, Vincenzo Roselli}
\institute{Dipartimento di Ingegneria, University Roma Tre, Italy\\
\email{\{frati,patrigna,roselli\}@dia.uniroma3.it}}
\maketitle

\begin{abstract}
We study a family of algorithms, introduced by Chan~{\em [SODA 1999]}, for drawing ordered rooted binary trees. Any algorithm in this family (which we name an {\em LR-algorithm}) takes in input an ordered rooted binary tree $T$ with a root $r_T$, and recursively constructs drawings $\Gamma_L$ of the left subtree $L$ of $r_T$ and $\Gamma_R$ of the right subtree $R$ of $r_T$; then either it applies the {\em left rule}, i.e., it places $\Gamma_L$ one unit below and to the left of $r_T$, and $\Gamma_R$ one unit below $\Gamma_L$ with the root of $R$ vertically aligned with $r_T$, or it applies the {\em right rule}, i.e., it places $\Gamma_R$ one unit below and to the right of $r_T$, and $\Gamma_L$ one unit below $\Gamma_R$ with the root of $L$ vertically aligned with $r_T$. In both cases, the edges between $r_T$ and its children are represented by straight-line segments. Different LR-algorithms result from different choices on whether the left or the right rule is applied at any non-leaf node of $T$. We are interested in constructing {\em LR-drawings} (that are drawings obtained via LR-algorithms) with small width. Chan showed three different LR-algorithms that achieve, for an ordered rooted binary tree with $n$ nodes, width $O(n^{0.695})$, width $O(n^{0.5})$, and width $O(n^{0.48})$. 

We prove that, for every $n$-node ordered rooted binary tree, an LR-drawing with minimum width can be constructed in $O(n^{1.48})$ time. Further, we show an infinite family of $n$-node ordered rooted binary trees requiring $\Omega(n^{0.418})$ width in any LR-drawing; no lower bound better than $\Omega(\log n)$ was previously known. Finally, we present the results of an experimental evaluation that allowed us to determine the minimum width of all the ordered rooted binary trees with up to $455$ nodes.

Our interest in LR-drawings is mainly motivated by a result of Di Battista and Frati {\em [Algorithmica 2009]}, who proved that $n$-vertex outerplanar graphs have outerplanar straight-line drawings in $O(n^{1.48})$ area by means of a drawing algorithm which resembles an LR-algorithm. 

We deepen the connection between LR-drawings and outerplanar straight-line drawings by proving that, if $n$-node ordered rooted binary trees have LR-drawings with $f(n)$ width, for any function $f(n)$, then $n$-vertex outerplanar graphs have outerplanar straight-line drawings in $O(f(n))$ area. 

Finally, we exploit a structural decomposition for ordered rooted binary trees introduced by Chan in order to prove that every $n$-vertex outerplanar graph has an outerplanar straight-line drawing in $O\left(n\cdot 2^{\sqrt{2 \log_2 n}} \sqrt{\log n}\right)$ area. 
\end{abstract}

\section{Introduction} \label{se:introduction}

In this paper we study algorithms for constructing geometric representations of ordered rooted binary trees. This research topic has been investigated for a long time, because of the importance and the ubiquitousness of ordered rooted binary trees in computer science. Geometric models for representing ordered rooted binary trees were already discussed almost 50 years ago in Knuth's foundational book ``The Art of Computer Programming''~\cite{k-acp-68}. We explicitly mention here the notorious Reingold and Tilford's algorithm~\cite{rt-tdt-81} (counting more than 570 citations, according to Google Scholar) and invite the reader to consult the survey by Rusu~\cite{r-tda-13} as a reference point for a plethora of other tree drawing algorithms.

We introduce some definitions. A {\em rooted tree} $T$ is a tree with one distinguished node called {\em root}, which we denote by $r_T$. For any node $s\neq r_T$ in $T$, the {\em parent} of $s$ is the neighbor of $s$ in the path between $s$ and $r_T$ in $T$; also, for any node $s$ in $T$, the {\em children} of $s$ are the neighbors of $s$ different from its parent. For any node $s\neq r_T$ in $T$, the {\em subtree} of $T$ rooted at $s$ is defined as follows: remove from $T$ the edge between $s$ and its parent, thus separating $T$ in two trees; the one containing $s$ is the subtree of $T$ rooted at $s$. A {\em rooted binary tree} is a rooted tree such that every node has at most two children. An {\em ordered rooted binary tree} $T$ is a rooted binary tree in which any node $s\neq r_T$ is either designated as the {\em left child} or as the {\em right child} of its parent, so that a node with two children has a left and a right child. The subtree of $T$ rooted at the left (right) child of a node $s$ is the {\em left} ({\em right}) {\em subtree} of $s$; we also call {\em left} ({\em right}) {\em subtree} of a path $P$ in $T$ any left (right) subtree of a node in $P$ whose root is not in $P$.

At the Tenth Symposium on Discrete Algorithms held in 1999, Chan~\cite{c-anlabdbt-99,c-nlabdbt-02} introduced a simple family of algorithms to draw ordered rooted binary trees; we name the algorithms in this family {\em LR-algorithms}. An LR-algorithm is defined as follows. Consider an ordered rooted binary tree $T$. If $T$ has one node, then represent it as a point in the plane. Otherwise, recursively construct drawings $\Gamma_L$ of the left subtree $L$ of $r_T$ and $\Gamma_R$ of the right subtree $R$ of $r_T$. Denote by $B(\Gamma)$ the {\em bounding box} of a drawing $\Gamma$, i.e., the smallest axis-parallel rectangle containing $\Gamma$ in the closure of its interior. Then apply either:

  \begin{figure}[htb]
    \centering
    \hfill
	\subfloat[]{
		\includegraphics[width=.15\textwidth]{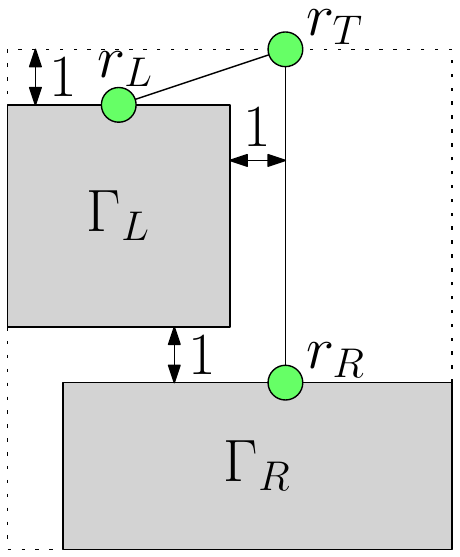}} 
    \hfill
	\subfloat[]{
		\includegraphics[width=.15\textwidth]{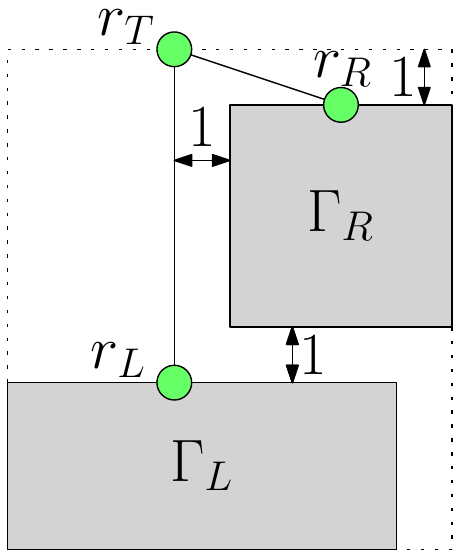}}
    \hfill \ 
	\caption{(a) Illustration for the left rule. (b) Illustration for the right rule.}
    \label{fig:lr-rules}
  \end{figure}

\begin{itemize}
\item the {\em left rule} (see Fig.~\ref{fig:lr-rules}(a)), i.e., place $\Gamma_L$ so that the top side of $B(\Gamma_L)$ is one unit below $r_T$ and so that the right side of $B(\Gamma_L)$ is one unit to the left of $r_T$, and place $\Gamma_R$ so that the top side of $B(\Gamma_R)$ is one unit below the bottom side of $B(\Gamma_L)$ and so that $r_R$ is vertically aligned with $r_T$; or
\item the {\em right rule} (see Fig.~\ref{fig:lr-rules}(b)), i.e., place $\Gamma_R$ so that the top side of $B(\Gamma_R)$ is one unit below $r_T$ and so that the left side of $B(\Gamma_R)$ is one unit to the right of $r_T$, and place $\Gamma_L$ so that the top side of $B(\Gamma_L)$ is one unit below the bottom side of $B(\Gamma_R)$ and so that $r_L$ is vertically aligned with $r_T$. 
\end{itemize}

By fixing different criteria for choosing whether to apply the left or the right rule at each internal node of $T$, one obtains different LR-algorithms. We call {\em LR-drawing} the output of an LR-algorithm. 

LR-drawings are a special class of {\em ideal drawings}, which constitute the main topic of investigation in Chan's paper~\cite{c-anlabdbt-99,c-nlabdbt-02} and are a very natural drawing standard for ordered rooted binary trees. They require the drawing to be: (i) {\em planar}, i.e., no two curves representing edges should cross -- this property helps to distinguish distinct edges; (ii) {\em straight-line}, i.e., each curve representing an edge is a straight-line segment -- this property helps to track an edge in the drawing; (iii) {\em strictly upward}, i.e., each node is below its parent -- this property helps to visualize the parent-child relationship between nodes; and (iv) {\em strongly order-preserving}, i.e., the left (right) child of a node is to the left (resp.\ right) or on the same vertical line of its parent -- this property allows to easily distinguish the left and right child of a node. 

As well-established in the graph drawing literature (see, e.g.,~\cite{BattistaETT99,kw-dgmm-01,NishizekiR04}), an optimization objective of primary importance for a drawing algorithm is to construct drawings with a {\em small area}. This is usually formalized by requiring the vertices to lie {\em in a grid}, that is, at points with integer coordinates, by  defining the {\em width} and {\em height} of $\Gamma$ as the number of grid columns and rows intersecting $\Gamma$, respectively\footnote{According to this definition, the width of $\Gamma$ is the geometric width of $B(\Gamma)$ plus one, and similar for the height.}, and by then defining the {\em area} of $\Gamma$ as its width times its height. 

Ideal drawings of $n$-node ordered rooted binary trees can be easily constructed in $O(n^2)$ area. For example, the width and the height of {\em any} LR-drawing are at most $n$ and exactly $n$, respectively. Because of the strictly-upward property, any ideal drawing of an $n$-node ordered rooted binary tree requires $\Omega(n)$ height if the tree contains a path with $\Omega(n)$ nodes from the root to a leaf. Thus, in order to construct ideal drawings with small area, the main goal is to minimize the width of the drawing. Chan exhibited several algorithms to construct ideal drawings. Three of them are in fact LR-algorithms that construct LR-drawings with $O(n^{0.695})$, $O(n^{0.5})$, and $O(n^{0.48})$ width, respectively. Better bounds than those resulting from LR-algorithms are however known for the width of ideal drawings. Namely, Garg and Rusu proved that every $n$-node ordered rooted binary tree has an ideal drawing with $O(\log n)$ width and $O(n\log n)$ area~\cite{gr-aeop-03}, which are the best possible bounds~\cite{cdp-noad-92}. Nevertheless, there are several reasons to study LR-drawings with small width and area.

\begin{wrapfigure}[23]{r}{0.1\textwidth}
	\centering
	\vspace{-8mm}
      \includegraphics[height=0.55\textwidth]{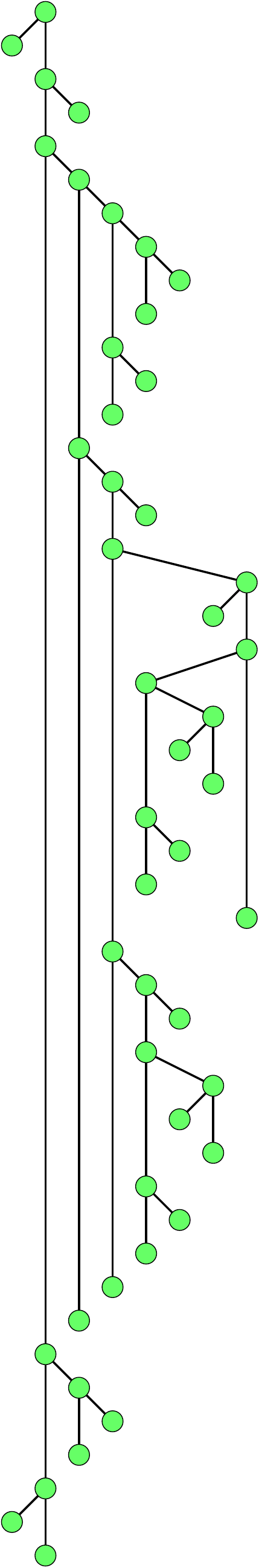}
    \caption{}
    \label{fi:ideal-t4}
\end{wrapfigure}

First, while one might design complicated schema to decide whether to apply the left or the right rule at any internal node of an ordered rooted binary tree, the geometric construction underlying an LR-algorithm is very easy to understand and implement. Second, as noted by Chan~\cite{c-anlabdbt-99,c-nlabdbt-02} an LR-drawing satisfies a number of additional geometric properties with respect to a general ideal drawing. For example, in an LR-drawing any two disjoint subtrees are separable by a horizontal line and any angle formed by the two edges between a node and its children is at least $\pi/4$. Third, let $w^*_T$ denote the minimum width of any LR-drawing of an ordered rooted binary tree $T$; also, let $w^*_n$ be the maximum value of $w^*_T$ among all the ordered rooted binary trees $T$ with $n$ nodes. In this paper we are interested in computing $w^*_T$ efficiently and in determining the asymptotic behavior of $w^*_n$. The value of $w^*_T$ obeys a natural recursive formula; namely $w^*_T=\min_{P} \{1+ \max_{L}\{w^*_{L}\} + \max_{R}\{w^*_{R}\}\}$, where the minimum is among all the paths $P$ starting at $r_T$, and the first and second maxima are among all the left and right subtrees of $P$, respectively\footnote{The intuition for this formula is that in any LR-drawing $\Gamma$ of $T$ a path $P$ starting at $r_T$ lies on a grid column $\ell$; thus the width of $\Gamma$ is the number of grid columns that intersect $\Gamma$ to the left of $\ell$ -- which is the maximum, among all the left subtrees $L$ of $P$, of the minimum width of an LR-drawing of $L$ -- plus the number of grid columns that intersect $\Gamma$ to the right of $\ell$ -- which is the maximum, among all the right subtrees $R$ of $P$, of the minimum width of an LR-drawing of $R$ -- plus one -- which corresponds to $\ell$.}. Our study of LR-drawings with small width might hence find application in problems (not necessarily related to graph drawing) in which a similar recurrence appears. Fourth and most importantly for this paper, LR-drawings with small width have a strong connection with outerplanar straight-line drawings of outerplanar graphs with small area, as will be described later.

In Section~\ref{se:trees} we prove that, for every $n$-node ordered rooted binary tree $T$, an LR-drawing of $T$ with minimum width $w^*_T$ (and with minimum area) can be constructed in $O(n \cdot w^*_T)\in O(n^{1.48})$ time. Chan~\cite{c-anlabdbt-99,c-nlabdbt-02} noted that ``By dynamic programming, one can compute in polynomial time the exact minimum area of'' any LR-drawing of $T$. Our sub-quadratic time bound is obtained by investigating the {\em representation sequence} of $T$, which is a sequence of $O(w^*_T)$ integers that conveys all the relevant information about the width of the LR-drawings of $T$. Further, we show that, for infinitely many values of $n$, there exists an $n$-node ordered rooted binary tree $T_h$ requiring $\Omega\left(n^{\frac{1}{\log_2 (3+\sqrt 5)}}\right )\in \Omega(n^{0.418})$ width in any LR-drawing; no lower bound better than $\Omega(\log n)$ was previously known~\cite{cdp-noad-92}. Since the height of any LR-drawing of an $n$-node tree is $n$, $T_h$ requires $\Omega(n^{1.418})$ area in any LR-drawing; hence near-linear area bounds cannot be achieved for LR-drawings, differently from general ideal drawings. Note that the exponents in these lower bounds are only $0.062$ apart from the corresponding upper bounds. Finally, we exploited again the concept of representation sequence in order to devise an experimental evaluation that determined the minimum width of all the ordered rooted binary trees with up to $455$ nodes. The most interesting outcome of this part of our research is perhaps the similarity of the trees that we have experimentally observed to require the largest width with the trees $T_h$ we defined for the lower bound. Fig.~\ref{fi:ideal-t4} shows a minimum-width LR-drawing of a smallest tree requiring width $8$ in any LR-drawing; this tree is also shown in Fig.~\ref{fig:lower-trees}(a). 

Section~\ref{se:outerplanar} deals with small-area drawings of outerplanar graphs. An {\em outerplanar graph} is a graph that excludes $K_4$ and $K_{2,3}$ as minors or, equivalently, a graph that admits an {\em outerplanar drawing}, that is a planar drawing in which all the vertices are incident to the outer face. Small-area outerplanar drawings have long been investigated. Biedl proved that every $n$-vertex outerplanar graph admits an outerplanar polyline drawing in $O(n\log n)$ area~\cite{b-smog-11}, where a {\em polyline} drawing represents each edge as a piece-wise linear curve. Garg and Rusu proved that every $n$-vertex outerplanar graph with maximum degree $d$ admits an outerplanar straight-line drawing in $O(d\cdot n^{1.48})$ area~\cite{gr-aepsdog-07}. The first sub-quadratic area upper bound for outerplanar straight-line drawings of $n$-vertex outerplanar graphs was established by Di Battista and Frati~\cite{bf-sadog-09}; the bound is $O(n^{1.48})$. Frati also proved an $O(d\cdot n\log n)$ area upper bound for outerplanar straight-line drawings of $n$-vertex outerplanar graphs with maximum degree $d$~\cite{f-sdog-12}. 


By looking at the $O(d\cdot n^{1.48})$ and $O(n^{1.48})$ area bounds above, it should come with no surprise that outerplanar straight-line drawings are related to LR-drawings of ordered rooted binary trees, for which the best known area upper bound is $O(n^{1.48})$~\cite{c-anlabdbt-99,c-nlabdbt-02}. We briefly describe the way this relationship was established in~\cite{bf-sadog-09}. Let $G$ be a maximal outerplanar graph with $n$ vertices and let $T$ be its {\em dual tree} ($T$ has a node for each internal face of $G$ and has an edge between two nodes if the corresponding faces of $G$ are adjacent). Di Battista and Frati~\cite{bf-sadog-09} proved that, if $T$ has a {\em star-shaped} drawing (which will be defined later) in a certain area, then $G$ has an outerplanar straight-line drawing in roughly the same area; they also showed how to construct a star-shaped drawing of $T$ in $O(n^{1.48})$ area; this algorithm is similar to an LR-algorithm, which is the reason why the $O(n^{1.48})$ bound arises.   

We prove that if an $n$-node ordered rooted binary tree $T$ has an LR-drawing with width $\omega$, then $T$ has a star-shaped drawing with width $O(\omega)$ (and area $O(n\cdot \omega))$. Our geometric construction is very similar to the one presented in~\cite{bf-sadog-09}, however it is enhanced so that no property other than the width bound\footnote{On the contrary, in order to prove the area bound for star-shaped drawings, \cite{bf-sadog-09} exploits a lemma from~\cite{c-anlabdbt-99,c-nlabdbt-02}, stating that, given any ordered rooted binary tree $T$, there exists a root-to-leaf path $P$ in $T$ such that, for any left subtree $\alpha$ and right subtree $\beta$ of $P$, $|\alpha|^{0.48} + |\beta|^{0.48} \leq (1-\delta)|T|^{0.48}$, for some constant $\delta > 0$.} is required to be satisfied by the LR-drawing of $T$ in order to ensure the existence of a star-shaped drawing  of $T$ with area $O(n\cdot \omega)$. 
Due to this result and to the relationship between the area requirements of star-shaped drawings and outerplanar straight-line drawings established in~\cite{bf-sadog-09}, any improvement on the $O(n^{0.48})$ width bound for LR-drawings of ordered rooted binary trees would imply an improvement on the $O(n^{1.48})$ area bound for outerplanar straight-line drawings of $n$-vertex outerplanar graphs. However, because of the lower bound for the width of LR-drawings proved in the first part of the paper, this approach cannot lead to the construction of outerplanar straight-line drawings of $n$-vertex outerplanar graphs in $o(n^{1.418})$ area. 

We prove that, for any constant $\varepsilon>0$, the $n$-vertex outerplanar graphs admit outerplanar straight-line drawings in $O(n^{1+\varepsilon})$ area. More precisely, our drawings have $O(n)$ height and $O(2^{\sqrt {2\log n}} \sqrt {\log n})$ width; the latter bound is smaller than any polynomial function of $n$. Hence, this establishes a near-linear area bound for outerplanar straight-line drawings of outerplanar graphs, improving upon the previously best known $O(n^{1.48})$ area bound~\cite{bf-sadog-09}. In order to achieve our result we exploit a structural decomposition for ordered rooted binary trees introduced by Chan~\cite{c-nlabdbt-02}, together with a quite complex geometric construction for star-shaped drawings of ordered rooted binary trees. 


\section{LR-Drawings of Ordered Rooted Binary Trees} \label{se:trees}

In this section we study LR-drawings of ordered rooted binary trees. 

\subsection{Representation sequences} \label{se:representations}

Our investigation starts by defining a combinatorial structure, called {\em representation sequence}, which can be associated to any ordered rooted binary tree $T$ and which conveys all the relevant information about the width of the LR-drawings of $T$. We first establish some preliminary properties and lemmata.

Consider an LR-drawing $\Gamma$ of an ordered rooted binary tree $T$. The {\em left width} of $\Gamma$ is the number of grid columns intersecting $\Gamma$ to the left of the grid column on which $r_T$ lies. The {\em right width} of $\Gamma$ is defined analogously. By definition of width, we have the following.

\begin{property} \label{pr:width}
The width of an LR-drawing $\Gamma$ is equal to its left width, plus its right width, plus one.	
\end{property}

For any $\alpha,\beta \in \mathbb{N}_0$, we say that a pair $(\alpha,\beta)$ is {\em feasible} for $T$ if $T$ admits an LR-drawing whose left width is at most $\alpha$ and whose right width is at most $\beta$. This definition implies the following.

\begin{property} \label{pr:feasible-dominance}
Consider an ordered rooted binary tree $T$. If a pair $(\alpha,\beta)$ is feasible for $T$, then every pair $(\alpha',\beta')$ with $\alpha',\beta' \in \mathbb{N}_0$, $\alpha'\geq \alpha$, and $\beta'\geq \beta$ is also feasible for $T$.
\end{property}

The next lemma will be used several times in the following.

\begin{lemma} \label{le:corner}
The pairs $(0,w^*_T)$ and $(w^*_T,0)$ are feasible for an ordered rooted binary tree $T$. 
\end{lemma}

\begin{proof}
We prove that the pair $(0,w^*_T)$ is feasible for $T$; the proof for the pair $(w^*_T,0)$ is symmetric.

The proof is by induction on the number $n$ of nodes of $T$. If $n=1$, then in any LR-drawing $\Gamma$ of $T$ both the left and the right width of $\Gamma$ are $0$, hence the pair $(0,0)$ is feasible for $T$. By Property~\ref{pr:feasible-dominance}, the pair $(0,1)$ is also feasible for $T$. This, together with $w^*_T=1$, implies the statement for $n=1$.

\begin{figure}[htb]
	\centering
	\includegraphics[height=.25\textwidth]{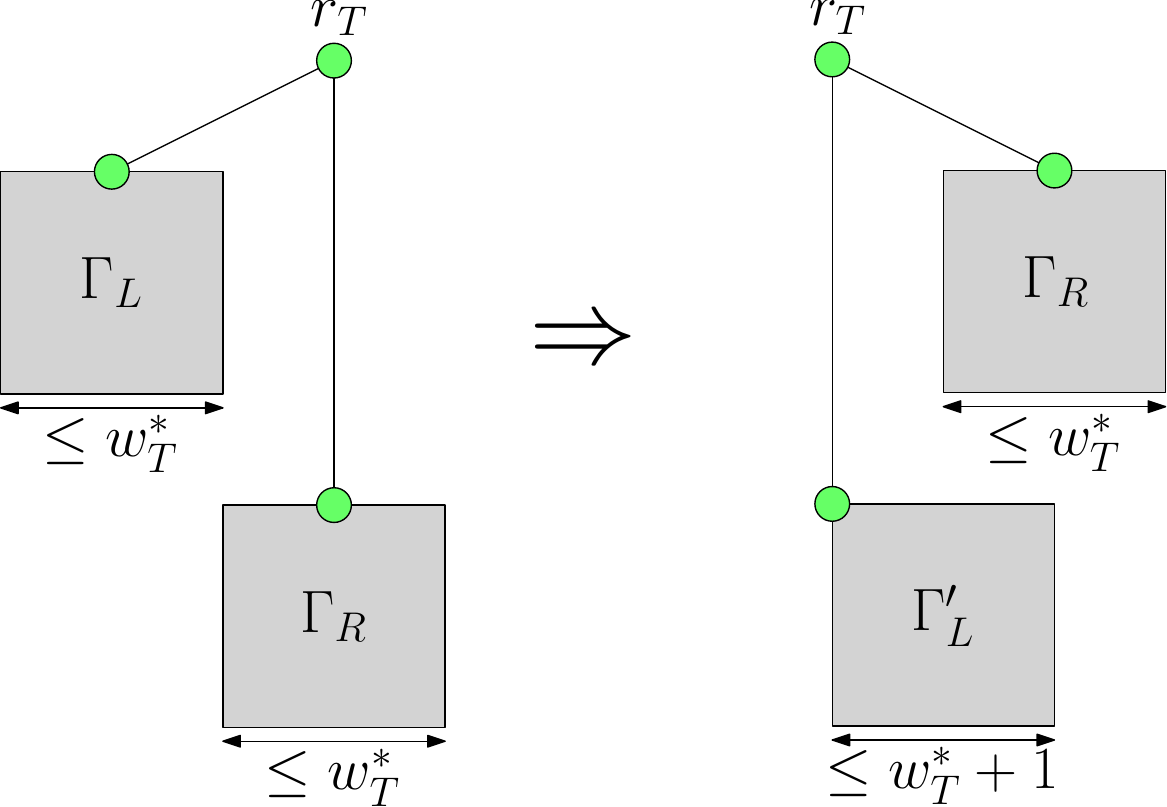}
	\caption{Illustration for the proof of Lemma.~\ref{le:corner}.}
	\label{fig:corner}
\end{figure}

If $n>1$, then assume that neither the left subtree $L$ nor the right subtree $R$ of $r_T$ is empty. The case in which $L$ or $R$ is empty is easier to handle. Refer to Fig.~\ref{fig:corner}. Consider any LR-drawing $\Gamma_T$ of $T$ with width $w^*_T$. Denote by $\Gamma_L$ and $\Gamma_R$ the LR-drawings of $L$ and $R$ in $\Gamma_T$, respectively. The width of each of $\Gamma_L$ and $\Gamma_R$ is at most $w^*_T$, given that the width of $\Gamma_T$ is $w^*_T$. Apply induction on $L$ to construct an LR-drawing $\Gamma'_L$ of $L$ with left width $0$ and right width at most $w^*_T$. Construct an LR-drawing $\Gamma'_T$ of $T$ by applying the right rule at $r_T$, while using $\Gamma_R$ as the LR-drawing of $R$ and $\Gamma'_L$ as the LR-drawing of $L$. Then the left width of $\Gamma'_T$ is equal to the left width of $\Gamma'_L$, hence it is $0$. Further, the right width of $\Gamma'_T$ is equal to the maximum between the width of $\Gamma_R$ and the right width of $\Gamma'_L$, which are both at most $w^*_T$; hence the pair $(0,w^*_T)$ is feasible for $T$.
\end{proof}

Property~\ref{pr:feasible-dominance} implies that there exists an infinite number of feasible pairs for $T$. Despite that, the set of feasible pairs for $T$ can be succinctly described by its {\em Pareto frontier}, which is the set of the feasible pairs $(\alpha,\beta)$ for $T$ such that no feasible pair $(\alpha',\beta')$ for $T$ exists with (i) $\alpha'< \alpha$ and $\beta'\leq \beta$ or (ii) $\alpha'\leq \alpha$ and $\beta'< \beta$.

More formally, the {\em representation sequence} of an ordered rooted binary tree $T$, which we denote by  ${\cal S}_T$, is an ordered list of integers (indexed by the numbers $0,1,2,\dots$) satisfying the following properties:

\begin{itemize} 
\item[(a)] the value ${\cal S}_T(i)$ of the element of ${\cal S}_T$ with index $i$ is the smallest integer $j$ such that $T$ admits an LR-drawing with left width at most $i$ and right width $j$; and
\item[(b)] the value of the second to last element of ${\cal S}_T$ is greater than $0$ and the value of the last element of ${\cal S}_T$ is equal to $0$.
\end{itemize} 

We let $k_T$ denote the number of elements in ${\cal S}_T$. Note that the values ${\cal S}_T(0),\dots,{\cal S}_T(k_T-1)$ in a representation sequence ${\cal S}_T$ are non-increasing, given that if a pair $(i,{\cal S}_T(i))$ is feasible for $T$, then the pair $(i+1,{\cal S}_T(i))$ is also feasible for $T$, by Property~\ref{pr:feasible-dominance}. For example, the tree $T_3$ shown in Fig.~\ref{fig:theory-lower-bound}(b) (which we use for the lower bound on the width of LR-drawings) has ${\cal S}_{T_3}=[6,5,5,3,3,1,0]$.

Note that, if $T$ is a root-to-leaf path, then ${\cal S}_T=[0]$, since $T$ has an LR-drawing in which all the nodes are on the same vertical line. Also, any complete binary tree $T$ with height $h+1$ (i.e., with $h+1$ nodes on any root-to-leaf path) has ${\cal S}_T=[h,\dots,h,0]$, where $h$ elements are equal to $h$. This is can be proved by induction and by the following lemma.

\begin{lemma} \label{le:exploding-structure}
Consider any ordered rooted binary tree $T$. Let $T'$ be the tree such that the left subtree $L$ and the right subtree $R$ of $r_{T'}$ are two copies of $T$. Then ${\cal S}_{T'}=[\underbrace{w^*_T}_{\textrm{index } 0},\dots,\underbrace{w^*_T}_{\textrm{index } w^*_T-1},\underbrace{0}_{\textrm{index } w^*_T}]$.
\end{lemma}

\begin{proof}
First, we prove that ${\cal S}_{T'}(i)= w^*_T$, for $i=0,\dots,w^*_T-1$. 

We prove that ${\cal S}_{T'}(i)\geq w^*_T$. Consider any LR-drawing $\Gamma_{T'}$ of $T'$ with left width $i\leq w^*_T-1$. If $\Gamma_{T'}$ used the left rule at $r_{T'}$, then the LR-drawing of $L$ in $\Gamma_{T'}$ would be entirely to the left of $r_{T'}$; hence, the left width of $\Gamma_{T'}$ would be at least $w^*_T$, while it is at most $i$, by assumption. It follows that $\Gamma_{T'}$ uses the right rule at $r_{T'}$ and the LR-drawing of $R$ in $\Gamma_{T'}$ is entirely to the right of $r_{T'}$; hence, ${\cal S}_{T'}(i)\geq w^*_T$. 

We prove that ${\cal S}_{T'}(i)\leq w^*_T$. Consider an LR-drawing $\Gamma_{R}$ of $R$ with width $w^*_T$, and an LR-drawing $\Gamma_{L}$ of $L$ with left width at most $i$ and right width $w^*_T$; $\Gamma_{L}$ exists since pair $(0,w^*_T)$ is feasible for $L$, by Lemma~\ref{le:corner}. Construct an LR-drawing $\Gamma_{T'}$ of $T'$ by applying the right rule at $r_{T'}$, while using $\Gamma_L$ and $\Gamma_R$ as LR-drawings for $L$ and $R$, respectively. Since $r_{T'}$ and $r_L$ are on the same vertical line, the left width of $\Gamma_{T'}$ is equal to the left width of $\Gamma_L$, which is at most $i$, and the right width of $\Gamma_{T'}$ is the maximum between the right width of $\Gamma_L$ and the width of $\Gamma_{R}$, which are both equal to $w^*_T$. Hence, ${\cal S}_{T'}(i)\leq w^*_T$. 

Finally, we prove that ${\cal S}_{T'}(w^*_T)= 0$. Consider an LR-drawing $\Gamma_{L}$ of $L$ with width at most $w^*_T$, and an LR-drawing $\Gamma_{R}$ of $R$ with left width at most $w^*_T$ and right width $0$; the latter drawing exists by Lemma~\ref{le:corner}. Construct an LR-drawing $\Gamma_{T'}$ of $T'$ by applying the left rule at $r_{T'}$, while using $\Gamma_L$ and $\Gamma_R$ as LR-drawings for $L$ and $R$, respectively. Since $r_{T'}$ and $r_R$ are on the same vertical line, the right width of $\Gamma_{T'}$ is equal to the right width of $\Gamma_R$, which is $0$, and the left width of $\Gamma_{T'}$ is the maximum between the left width of $\Gamma_R$ and the width of $\Gamma_{L}$, which are both at most $w^*_T$. Hence, ${\cal S}_{T'}(w^*_T)= 0$. 
\end{proof}

As a final lemma of this section we bound the number of elements in a representation sequence. 

\begin{lemma} \label{le:length-bound}
Consider any ordered rooted binary tree $T$. Then the length $k_T$ of ${\cal S}_T$ is either $w^*_T$ or $w^*_T+1$. 
\end{lemma}

\begin{proof}
First, $k_T\leq w^*_T-1$ would imply that the last element of ${\cal S}_T$ has index less than or equal to $w^*_T-2$ and value $0$. By Property~\ref{pr:width}, there would exist an LR-drawing of $T$ with width at most $w^*_T-2 +0 + 1<w^*_T$, which is not possible by definition of $w^*_T$. It follows that $k_T\geq w^*_T$.

Second, Lemma~\ref{le:corner} implies that the pair $(w^*_T,0)$ is feasible for $T$, hence $k_T=w^*_T$ or $k_T=w^*_T+1$, depending on whether the pair $(w^*_T-1,0)$ is feasible for $T$ or not.
\end{proof}

\subsection{Algorithms for Optimal LR-drawings} \label{se:algorithms-LR}

There are two main reasons to study the representation sequence ${\cal S}_T$ of an ordered rooted binary tree $T$. The first one is that the minimum width among all the LR-drawings of $T$ can be easily retrieved from ${\cal S}_T$; the second one is that ${\cal S}_T$ can be easily constructed starting from the representation sequences of the subtrees of $r_T$. The next lemmata formalize these claims.

\begin{lemma} \label{le:sequence-width}
For any ordered rooted binary tree $T$, the minimum width among all the LR-drawings of $T$ is equal to $\min_{i=0}^{k_T-1} \{i+{\cal S}_T(i)+1\}$.
\end{lemma}
	
\begin{proof}
Consider any LR-drawing $\Gamma$ of $T$ with minimum width $w^*_T$, and let $\alpha$ and $\beta$ be the left and right width of $\Gamma$, respectively. By Property~\ref{pr:width}, we have that $w^*_T=\alpha+\beta+1$. By definition of ${\cal S}_T$, we have that ${\cal S}_T(\alpha)\leq \beta$. Finally, by the minimality of $w^*_T$ we have ${\cal S}_T(\alpha)= \beta$, which proves the statement.
\end{proof}

\begin{lemma} \label{le:sequence-construction}
Let $T$ be an ordered rooted binary tree. Let $L$ and $R$ be the (possibly empty) left and right subtrees of $r_T$, respectively. The following statements hold true.

\begin{itemize} 
\item If $L$ and $R$ are both empty, then ${\cal S}_T=[0]$. 
\item If $L$ is empty and $R$ is not, then ${\cal S}_T={\cal S}_R$. 
\item If $R$ is empty and $L$ is not, then ${\cal S}_T={\cal S}_L$. 
\item Finally, if neither $L$ nor $R$ is empty, then 
$$
{\cal S}_T= [\underbrace{\max \{ {\cal S}_L(0),w^*_R\}}_{ \textrm{index } 0}, \dots, \underbrace{\max \{ {\cal S}_L(w^*_L-1),w^*_R\}}_{ \textrm{index } w^*_L-1},\underbrace{{\cal S}_R(w^*_L)}_{ \textrm{index } w^*_L},\dots,\underbrace{{\cal S}_R(k_R -1)}_{ \textrm{index } k_R -1}].
$$
\end{itemize} 
\end{lemma}

\begin{proof}
We distinguish four cases, based on whether $L$ and $R$ are empty or not.

\begin{itemize}
\item If both $L$ and $R$ are empty, then $T$ consists of a single node, hence there is only one LR-drawing $\Gamma$ of $T$; both the left and the right width of $\Gamma$ are $0$, hence ${\cal S}_T=[0]$. 

\item If $L$ is empty and $R$ is not, we prove that ${\cal S}_T(i)={\cal S}_R(i)$, for any $i=0,\dots,k_R-1$. 

First, we prove that ${\cal S}_T(i)\leq {\cal S}_R(i)$. Consider an LR-drawing $\Gamma_R$ of $R$ with left width at most $i$ and right width ${\cal S}_R(i)$. Construct an LR-drawing $\Gamma_T$ of $T$ by applying the left rule at $r_T$, while using $\Gamma_R$ as the LR-drawing of $R$. Since $r_T$ and $r_R$ are on the same vertical line, the left (right) width of $\Gamma_T$ is equal to the left (resp.\ right) width of $\Gamma_R$, which is at most $i$ (resp.\ which is ${\cal S}_R(i)$). Hence, ${\cal S}_T(i)\leq {\cal S}_R(i)$.

Second, we prove that ${\cal S}_R(i)\leq {\cal S}_T(i)$. Consider an LR-drawing $\Gamma_T$ of $T$ with left width at most $i$ and right width ${\cal S}_T(i)$; denote by $\Gamma_R$ the LR-drawing of $R$ in $\Gamma_T$. If $\Gamma_T$ uses the left rule at $r_T$, then $r_T$ and $r_R$ are on the same vertical line; then the left (right) width of $\Gamma_R$ is equal to the left (resp.\ right) width of $\Gamma_T$, which is at most $i$ (resp.\ which is ${\cal S}_T(i)$). Hence, ${\cal S}_R(i)\leq {\cal S}_T(i)$. If $\Gamma_T$ uses the right rule at $r_T$, then $\Gamma_R$ is entirely to the right of $r_T$, hence ${\cal S}_T(i)= w^*_R$. By Lemma~\ref{le:corner}, the pair $(0,w^*_R)$ is feasible for $R$, hence ${\cal S}_R(i)\leq w^*_R$. Hence, ${\cal S}_R(i)\leq {\cal S}_T(i)$. 

\item If $R$ is empty and $L$ is not, the discussion is symmetric to the one for the previous case.
 
\item Finally, assume that neither $L$ nor $R$ is empty. In order to compute the value of ${\cal S}_T(i)$, we distinguish the case in which $i\leq w^*_L-1$ from the one in which $i\geq w^*_L$.

\begin{itemize}
\item Suppose first that $i\leq  w^*_L-1$; we prove that ${\cal S}_T(i) = \max \{ {\cal S}_L(i),w^*_R\}$.

First, we prove that ${\cal S}_T(i)\leq \max \{ {\cal S}_L(i),w^*_R\}$. Consider an LR-drawing $\Gamma_L$ of $L$ with left width at most $i$ and right width ${\cal S}_L(i)$. Also, consider an LR-drawing $\Gamma_R$ of $R$ with width $w^*_R$. Construct an LR-drawing $\Gamma_T$ of $T$ by applying the right rule at $r_T$, while using $\Gamma_L$ and $\Gamma_R$ as LR-drawings for $L$ and $R$, respectively. Since $r_T$ and $r_L$ are on the same vertical line, the left width of $\Gamma_T$ is equal to the left width of $\Gamma_L$, which is at most $i$, and the right width of $\Gamma_T$ is equal to the maximum between the right width of $\Gamma_L$ and $w^*_R$. Hence, ${\cal S}_T(i)\leq \max \{ {\cal S}_L(i),w^*_R\}$.

Second, we prove that ${\cal S}_T(i)\geq \max \{ {\cal S}_L(i),w^*_R\}$. Consider any LR-drawing $\Gamma_T$ of $T$ with left width at most $i$ and right width ${\cal S}_T(i)$. We have that $\Gamma_T$ uses the right rule at $r_T$. Indeed, if $\Gamma_T$ used the left rule at $r_T$, then the LR-drawing of $L$ in $\Gamma_T$ would be entirely to the left of $r_T$; hence, the left width of $\Gamma_T$ would be at least $w^*_L$, while it is at most $i$, by assumption. Since $\Gamma_T$ uses the right rule at $r_T$, the LR-drawing of $R$ in $\Gamma_T$ is entirely to the right of $r_T$, hence ${\cal S}_T(i)\geq w^*_R$. Further, $r_T$ and $r_L$ are on the same vertical line, thus the LR-drawing of $L$ in $\Gamma_T$ has left width at most $i$, and hence right width at least ${\cal S}_L(i)$; this implies that ${\cal S}_T(i)\geq {\cal S}_L(i)$. 

\item Suppose next that $i\geq  w^*_L$; we prove that ${\cal S}_T(i) = {\cal S}_R(i)$.

First, we prove that ${\cal S}_T(i)\leq {\cal S}_R(i)$. Consider an LR-drawing $\Gamma_L$ of $L$ with width $w^*_L$. Also, consider an LR-drawing $\Gamma_R$ of $R$ with left width at most $i$ and right width ${\cal S}_R(i)$. Construct an LR-drawing $\Gamma_T$ of $T$ by applying the left rule at $r_T$, while using $\Gamma_L$ and $\Gamma_R$ as LR-drawings for $L$ and $R$, respectively. Since $r_T$ and $r_R$ are on the same vertical line, the right width of $\Gamma_T$ is equal to the right width of $\Gamma_R$, which is ${\cal S}_R(i)$, and the left width of $\Gamma_T$ is equal to the maximum between $w^*_L$ and the left width of $\Gamma_R$; since $w^*_L$ and the left width of $\Gamma_R$ are both at most $i$, we have ${\cal S}_T(i)\leq {\cal S}_R(i)$.

Second, we prove that ${\cal S}_R(i)\leq {\cal S}_T(i)$. Consider any LR-drawing $\Gamma_T$ of $T$ with left width at most $i$. If $\Gamma_T$ uses the left rule at $r_T$, then $r_T$ and $r_R$ are on the same vertical line, thus the LR-drawing of $R$ in $\Gamma_T$ has left width at most $i$ and right width at most ${\cal S}_T(i)$. It follows that ${\cal S}_R(i)\leq {\cal S}_T(i)$. If $\Gamma_T$ uses the right rule at $r_T$, then the LR-drawing of $R$ in $\Gamma_T$ is entirely to the right of $r_T$, hence ${\cal S}_T(i)\geq w^*_R$. By Lemma~\ref{le:corner}, the pair $(0,w^*_R)$ is feasible for $R$, hence, ${\cal S}_R(i)\leq w^*_R$. It follows that ${\cal S}_R(i)\leq {\cal S}_T(i)$.
\end{itemize}
\end{itemize}

This concludes the proof. 
\end{proof}

We are now ready to show that the representation sequence of an ordered rooted binary tree $T$, and consequently the minimum width and area of any LR-drawing of $T$, can be computed efficiently.

\begin{theorem} \label{thm:compute-matrix}
The representation sequence of an $n$-node ordered rooted binary tree $T$ can be computed in $O(n\cdot w^*_T)\in O(n\cdot w^*_n)\in  O(n^{1.48})$ time. Further, an LR-drawing with minimum width can be constructed in the same time.
\end{theorem}

\begin{proof}
We compute the representation sequence associated to each subtree $T'$ of $T$ (and the value $w^*_{T'}$) by means of a bottom-up traversal of $T$. If $T'$ is a single node, then ${\cal S}_{T'}=[0]$ and $w^*_{T'}=1$. If $T'$ is not a single node, then assume that the representation sequences associated to the subtrees of $r_{T'}$ have already been computed. By Lemma~\ref{le:sequence-construction}, the value ${\cal S}_{T'}(i)$ can be computed in $O(1)$ time by the formula $\max \{ {\cal S}_L(i),w^*_R\}$ if $0\leq i\leq w^*_L-1$, or by the formula ${\cal S}_R(i)$ if $w^*_L\leq i \leq k_R -1$. Further, by Lemma~\ref{le:length-bound} the representation sequence ${\cal S}_{T'}$ has $O(w^*_{T'})\in  O(w^*_{T})$ entries, hence it can be computed in $O(w^*_{T'})\in  O(w^*_{T})$ time; the value $w^*_{T'}$ can also be computed in $O(w^*_{T'})$ time from ${\cal S}_{T'}$ as in Lemma~\ref{le:sequence-width}. Summing the $O(w^*_{T})$ bound up over the $n$ nodes of $T$ gives the $O(n\cdot w^*_T)$ bound. The bounds $O(n\cdot w^*_n)$ and $O(n^{1.48})$ respectively follow from the fact that $w^*_T \leq w^*_n$, by definition, and $w^*_n \in O(n^{0.48})$, by the results of Chan~\cite{c-nlabdbt-02}. 

Once the representation sequence for each subtree of $T$ has been computed, an LR-drawing $\Gamma_T$ of $T$ with width $w^*_T$ can be constructed in $O(n\cdot w^*_T)$ time by means of a top-down traversal of $T$. First, find a pair $(\alpha_T,\beta_T)$ such that $\alpha_T+\beta_T+1=w^*_T$ and such that ${\cal S}_{T}(\alpha_T)=\beta_T$. This pair exists and can be found in $O(w^*_T)$ time by Lemma~\ref{le:sequence-width}. Further, let $x(r_T)=0$ and $y(r_T)=0$. 

Now assume that, for some subtree $T'$ of $T$ (initially $T'=T$), a quadruple $( \alpha_{T'},\beta_{T'},x(r_{T'}),y(r_{T'}))$ has been associated to $T'$, where $\alpha_{T'}$ and $\beta_{T'}$ represent the left and right width of an LR-drawing $\Gamma_{T'}$ of $T'$ we aim to construct, respectively, and $x(r_{T'})$ and $y(r_{T'})$ are the coordinates of $r_{T'}$ in $\Gamma_{T'}$. Let $L$ and $R$ be the left and right subtrees of $r_{T'}$, respectively. 

\begin{itemize}
	\item If $w^*_{L}\leq \alpha_{T'}$, then the left rule is used at $r_{T'}$ to construct $\Gamma_{T'}$. Find a pair $(\alpha_{L},\beta_{L})$ satisfying $\alpha_{L}+\beta_{L}+1=w^*_{L}$ and ${\cal S}_{L}(\alpha_{L})=\beta_{L}$. This pair exists and can be found in $O(w^*_{L})\in O(w^*_T)$ time by Lemma~\ref{le:sequence-width}. Let $x(r_{L})=x(r_{T'})-\beta_{L}-1$ and $y(r_{L})=y(r_{T'})-1$. Visit $L$ with quadruple $(\alpha_{L},\beta_{L},x(r_{L}),y(r_{L}))$ associated to it; also, let $\alpha_{R}=\alpha_{T'}$, $\beta_{R}=\beta_{T'}$, $x(r_{R})=x(r_{T'})$, and $y(r_{R})=y(r_{T'})-|L|-1$. Visit $R$ with quadruple $(\alpha_{R},\beta_{R},x(r_{R}),y(r_{R}))$ associated to it.	
	\item If $w^*_{L}>\alpha_{T'}$, then the right rule is used at $r_{T'}$ to construct $\Gamma_{T'}$. Find a pair $(\alpha_{R},\beta_{R})$ satisfying $\alpha_{R}+\beta_{R}+1=w^*_{R}$ and ${\cal S}_{R}(\alpha_{R})=\beta_{R}$. This pair exists and can be found in $O(w^*_{R})\in O(w^*_T)$ time by Lemma~\ref{le:sequence-width}. Let $x(r_{R})=x(r_{T'})+\alpha_{R}+1$ and $y(r_{R})=y(r_{T'})-1$. Visit $R$ with quadruple $(\alpha_{R},\beta_{R},x(r_{R}),y(r_{R}))$ associated to it; also, let $\alpha_{L}=\alpha_{T'}$, $\beta_{L}=\beta_{T'}$, $x(r_{L})=x(r_{T'})$, and $y(r_{L})=y(r_{T'})-|R|-1$. Visit $L$ with quadruple $(\alpha_{L},\beta_{L},x(r_{L}),y(r_{L}))$ associated to it.
\end{itemize}

The correctness of the algorithm comes from Lemma~\ref{le:sequence-construction} (and its proof). The $O(n\cdot w^*_T)$ running time comes from the fact that the algorithm uses $O(w^*_T)$ time at each node of $T$.
\end{proof}

\begin{corollary} \label{cor:area-lr}
A minimum-area LR-drawing of an $n$-node ordered rooted binary tree $T$ can be constructed in $O(n\cdot w^*_T)\in O(n\cdot w^*_n)\in  O(n^{1.48})$ time.
\end{corollary}

\begin{proof}
Since any LR-drawing has height exactly $n$, the statement follows from Theorem~\ref{thm:compute-matrix}.
\end{proof}


\subsection{A Polynomial Lower Bound for the Width of LR-drawings} \label{se:poly-lower}

We describe an infinite family of ordered rooted binary trees $T_h$ that require large width in any LR-drawing. In order to do that, we first define an infinite family of sequences of integers. Sequence $\sigma_1$ consists of the integer $1$ only; for any $\ell>1$, sequence $\sigma_{\ell}$ is composed of two copies of $\sigma_{{\ell}-1}$ separated by the integer ${\ell}$, that is, $\sigma_{\ell} = \sigma_{{\ell}-1},{\ell},\sigma_{{\ell}-1}$. Thus, for example, $\sigma_4=1,2,1,3,1,2,1,4,1,2,1,3,1,2,1$. For $i=1,\dots,2^{\ell}-1$, we denote by $\sigma_{\ell}(i)$ the $i$-th term of $\sigma_{\ell}$. While here we defined $\sigma_{\ell}$ as a finite sequence with length $2^{\ell}-1$, the infinite sequence $\sigma_{\ell}$ with ${\ell}\rightarrow \infty$ is well-known and called {\em ruler function}: The $i$-th term of the sequence is the exponent of the largest power of $2$ which divides $2i$. See entry A001511 in the Encyclopedia of Integer Sequences~\cite{s-is-74}.

\begin{figure}[htb]
	\captionsetup[subfigure]{labelformat=empty}
\centering
\subfloat[]{
\includegraphics[height=.36\textwidth]{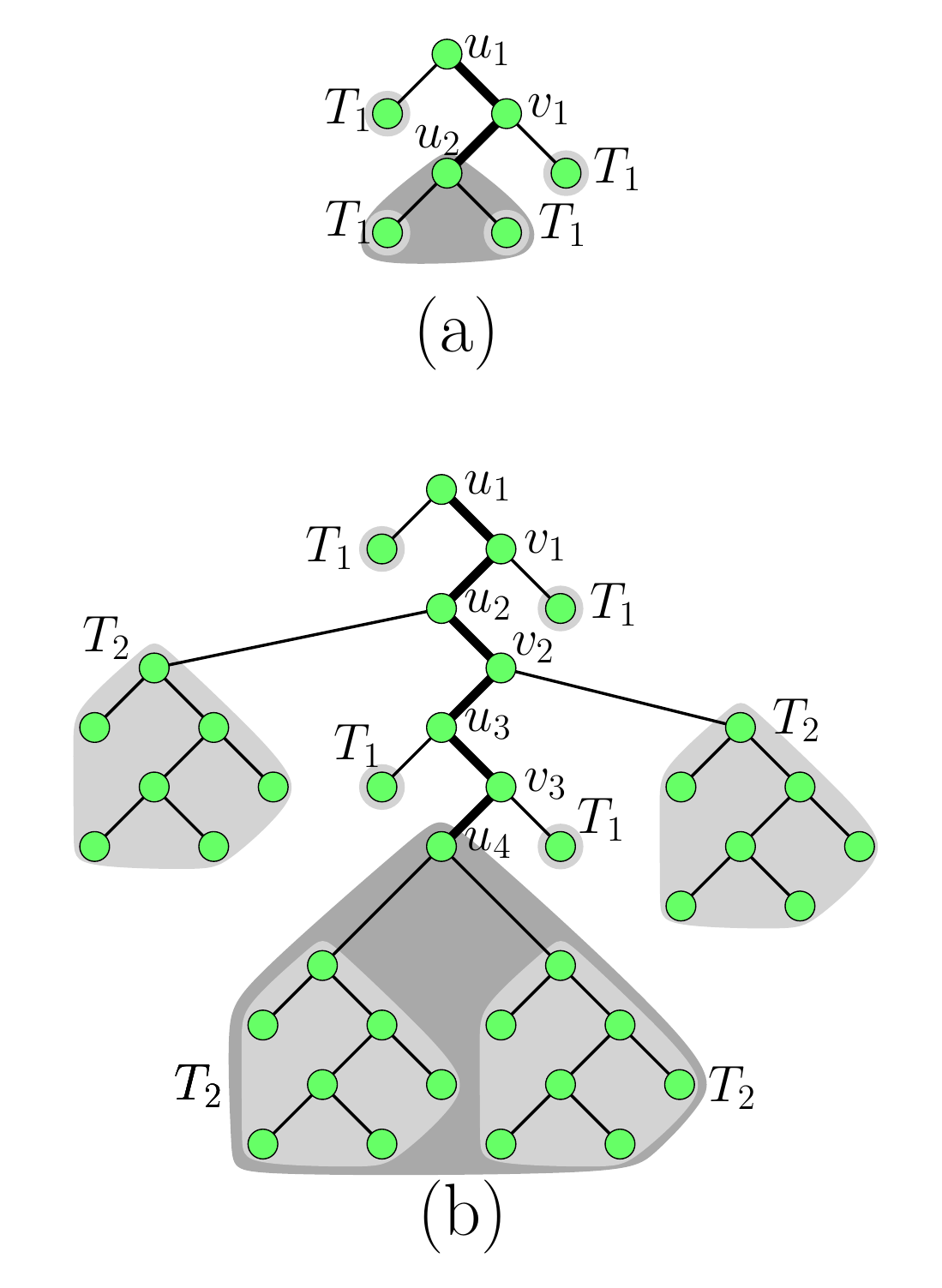}}
\subfloat[]{
\includegraphics[height=.36\textwidth]{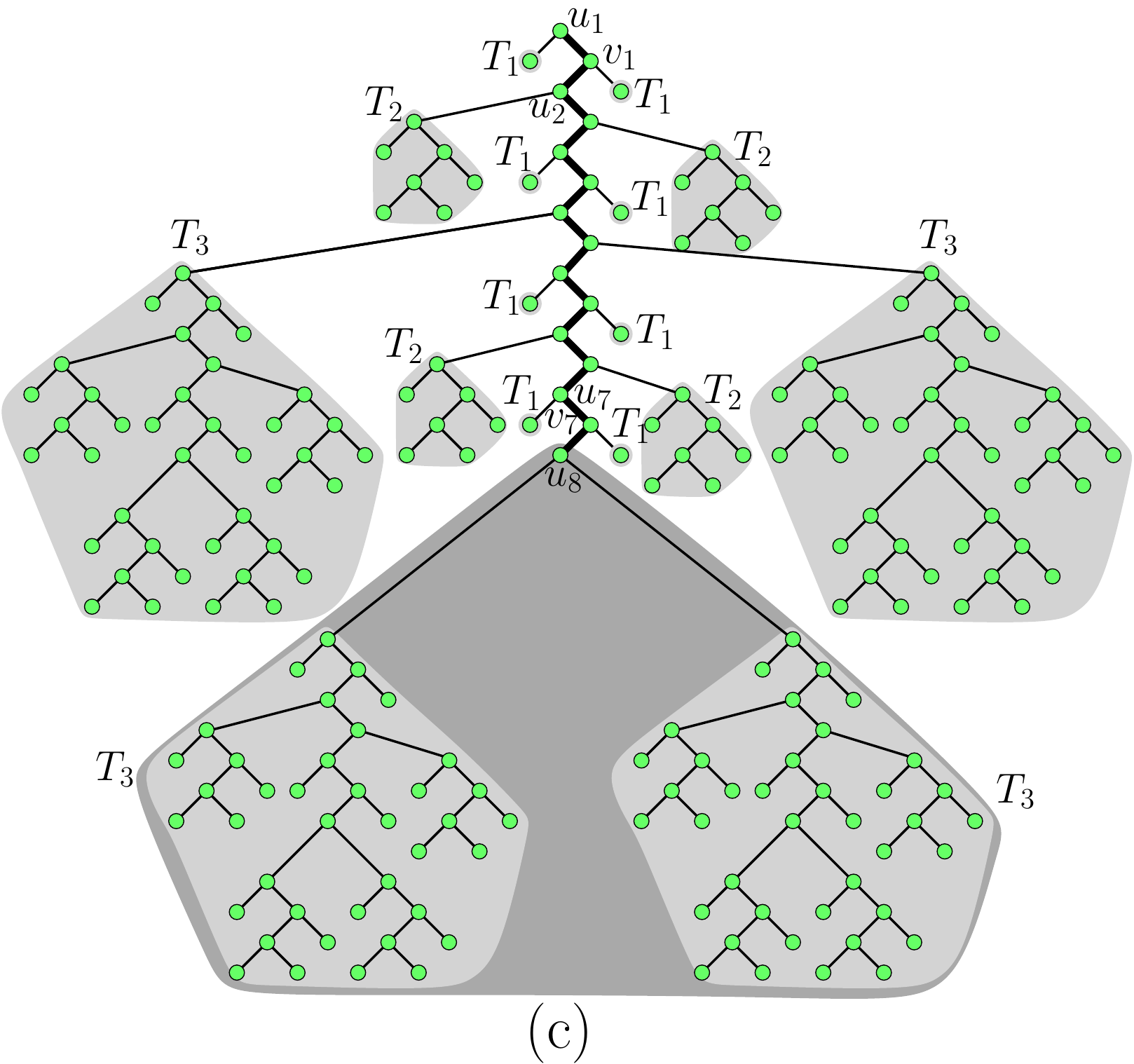}}
\subfloat[]{
\includegraphics[height=.36\textwidth]{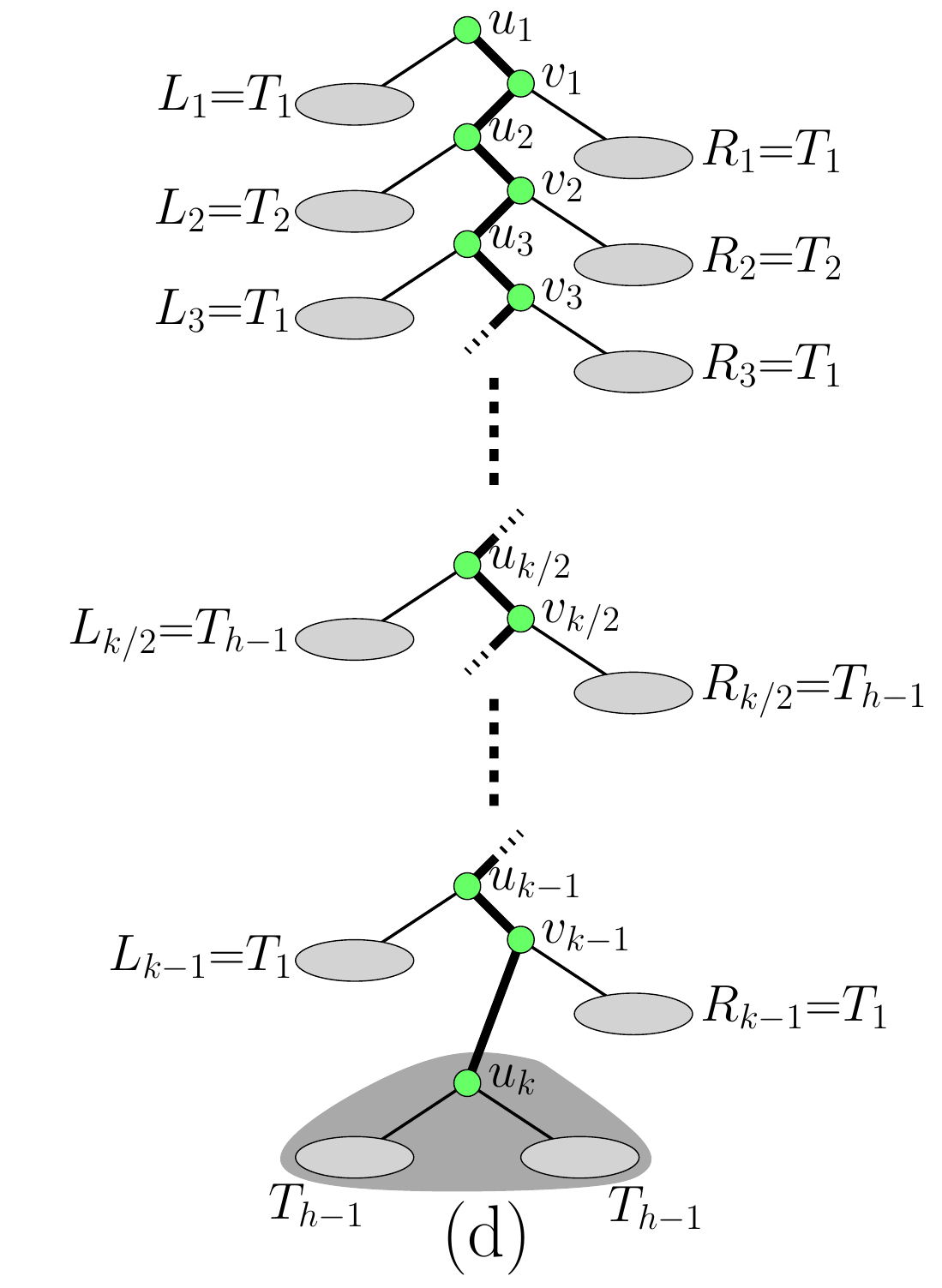}}
\caption{Illustration for Theorem~\ref{th:LR}. (a) $T_2$. (b) $T_3$. (c) $T_4$. (d) $T_h$.}
\label{fig:theory-lower-bound}
\end{figure}

We now describe the recursive construction of $T_h$. Tree $T_1$ consists of a single node. If $h>1$, tree $T_h$ is defined as follows (refer to Fig.~\ref{fig:theory-lower-bound}). First, $T_h$ contains a path $(u_1,v_1,u_2,v_2,\dots,u_{k-1},v_{k-1},u_{k})$ with $2^{h}-1$ nodes (note that $k=2^{h-1}$), where $u_1$ is the root of $T_h$; for $i=1,\dots,k-1$, node $v_i$ is the right child of $u_i$ and node $u_{i+1}$ is the left child of $v_i$. Further, take two copies of $T_{h-1}$ and let them be the left and right subtrees of $u_k$, respectively. Finally, for $i=1,\dots,k-1$, take two copies of $T_{\sigma_{h-1}(i)}$ and let them be the left subtree $L_i$ of $u_i$ and the right subtree $R_i$ of $v_i$, respectively. In the next two lemmata, we prove that tree $T_h$ requires a ``large width'' in any LR-drawing and that it has ``few'' nodes.


\begin{lemma} \label{le:LR-width}
The width of any LR-drawing of $T_h$ is at least $2^h-1$.
\end{lemma} 

\begin{proof}
The proof is by induction on $h$. The base case $h=1$ is trivial. 

In order to discuss the inductive case, we define another infinite family of sequences of integers, which we denote by $\pi_\ell$. Sequence $\pi_1$ consists of the integer $1$ only; for any $\ell>1$, we have $\pi_\ell = \pi_{\ell-1},2^\ell-1,\pi_{\ell-1}$. Thus, for example, $\pi_4=1,3,1,7,1,3,1,15,1,3,1,7,1,3,1$. For $i=1,\dots,2^{\ell}-1$, we denote by $\pi_\ell(i)$ the $i$-th element of $\pi_\ell$. The infinite sequence $\pi_\ell$ with $\ell\rightarrow \infty$ is well-known: The $i$-th term of the sequence is equal to $2^{x+1}-1$, where $x$ is the exponent of the largest power of $2$ which divides $i$. See entry A038712 in the Encyclopedia of Integer Sequences~\cite{s-is-74}.     

While sequence $\sigma_{h-1}$ was used for the construction of $T_h$ (recall that $L_i$ and $R_i$ are two copies of $T_{\sigma_{h-1}(i)}$), sequence $\pi_{h-1}$ is useful for the study of the minimum width of an LR-drawing of $T_h$. Indeed, by induction any LR-drawing of $L_i$ requires width $2^{\sigma_{h-1}(i)}-1$, which is equal to $\pi_{h-1}(i)$. Hence, the widths required by $L_1,\dots,L_{k-1}$ are $\pi_{h-1}(1),\dots,\pi_{h-1}(k-1)$, respectively; that is, they form the sequence $\pi_{h-1}$. A similar statement holds true for $R_1,\dots,R_{k-1}$. We are going to exploit the following.

\begin{property} \label{pr:sequence-max}
Let $\ell$ and $x$ be integers such that $\ell \geq 1$ and $1\leq x\leq 2^{\ell}-1$. For any $x$ consecutive elements in $\pi_{\ell}$, there exists one whose value is at least $x$.  
\end{property} 

\begin{proof}
We prove the statement by induction on ${\ell}$. If ${\ell}=1$, then $x=1$ and the statement follows since $\pi_1(1)=1$. Now assume that ${\ell}>1$ and consider any $x$ consecutive elements in $\pi_{\ell}$. Recall that $\pi_{\ell} = \pi_{\ell-1},2^\ell-1,\pi_{\ell-1}$. If all the $x$ elements belong to the first repetition of $\pi_{\ell-1}$ in $\pi_{\ell}$, or if all the $x$ elements belong to the second repetition of $\pi_{\ell-1}$ in $\pi_{\ell}$, then $x\leq 2^{\ell-1}-1$ and the statement follows by induction. Otherwise, since the $x$ elements are consecutive, the ``central'' element whose value is $2^{\ell}-1$ is among them. Then the statement follows since $x\leq 2^{\ell}-1$.
\end{proof}

We are now ready to discuss the inductive case of the lemma. Consider the subtrees $T(u_1),\dots,T(u_k)$ of $T_h$ rooted at $u_1,\dots,u_k$, respectively (note that $T(u_1)=T_h$). We claim that $T(u_j)$ requires width $2^{h-1}+k-j$ in any LR-drawing, for $j=1,\dots,k$. The lemma follows from the claim, as the latter (with $j=1$) implies that $T_h$ requires width $2^{h-1}+k-1= 2^h-1$ in any LR-drawing. 



Assume, for a contradiction, that the claim is not true, and let $j\in \{1,\dots,k\}$ be the maximum index such that there exists an LR-drawing $\Gamma$ of $T(u_j)$ whose width is less than $2^{h-1}+k-j$. First, since the subtrees of $u_k$ are two copies of $T_{h-1}$ and since by the inductive hypothesis $T_{h-1}$ requires width $2^{h-1}-1$ in any LR-drawing, by Lemma~\ref{le:exploding-structure} the representation sequence of $T(u_k)$ is 

$${\cal S}_{T(u_k)}=[\underbrace{2^{h-1}-1}_{\textrm{index } 0},\dots,\underbrace{2^{h-1}-1}_{\textrm{index } 2^{h-1}-2},\underbrace{0}_{\textrm{index } 2^{h-1}-1}].$$

Hence, $T(u_k)$ requires width $2^{h-1}$ in any LR-drawing, which implies that $j<k$. Let $\alpha$ and $\beta$ be the left and right width of $\Gamma$, respectively. In order to derive a contradiction, we prove that $\alpha+\beta+1\geq 2^{h-1}+k-j$. 

Suppose first (refer to Fig.~\ref{fig:proof-lower}(a)) that $\Gamma$ is constructed by using the left rule at $u_j,\dots,u_{k-1}$ and the right rule at $v_j,\dots,v_{k-1}$, hence nodes $u_j,\dots,u_{k-1},u_k$ and $v_j,\dots,v_{k-1}$ are all aligned on the same vertical line. Then $\alpha$ ($\beta$) is larger than or equal to the widths of $L_j,\dots,L_{k-1}$ (resp.\ of $R_j,\dots,R_{k-1}$) in $\Gamma$. We prove that $\alpha\geq 2^{h-1}-1$ or $\beta\geq 2^{h-1}-1$. If $\Gamma$ has left width $\alpha\leq 2^{h-1}-2$, then the LR-drawing of $T(u_k)$ in $\Gamma$ also has left width at most $2^{h-1}-2$, given that $u_j$ and $u_k$ are vertically aligned; since ${\cal S}_{T(u_k)}(2^{h-1}-2)= 2^{h-1}-1$, it follows that the right width of the LR-drawing of $T(u_k)$ in $\Gamma$ is at least $2^{h-1}-1$, and $\Gamma$ has right width $\beta\geq 2^{h-1}-1$. This proves that $\alpha\geq 2^{h-1}-1$ or $\beta\geq 2^{h-1}-1$. Assume that $\alpha\geq 2^{h-1}-1$, as the case $\beta\geq 2^{h-1}-1$ is symmetric. By induction, the width of the drawing of $R_i$ in $\Gamma$ is at least $\pi_{h-1}(i)$. Hence, the widths of the subtrees $R_j,\dots,R_{k-1}$ form a sequence of $k-j\geq 1$ consecutive elements of $\pi_{h-1}$. By Property~\ref{pr:sequence-max}, there exists an element $\pi_{h-1}(i)$ whose value is at least $k-j$. Then $\beta\geq k-j$ and $\alpha+\beta+1\geq (2^{h-1}-1) + (k-j) +1 = 2^{h-1}+k-j$, a contradiction.

\begin{figure}[htb]
	\centering
	\hfill
	\subfloat[]{
		\includegraphics[height=.45\textwidth,page=1]{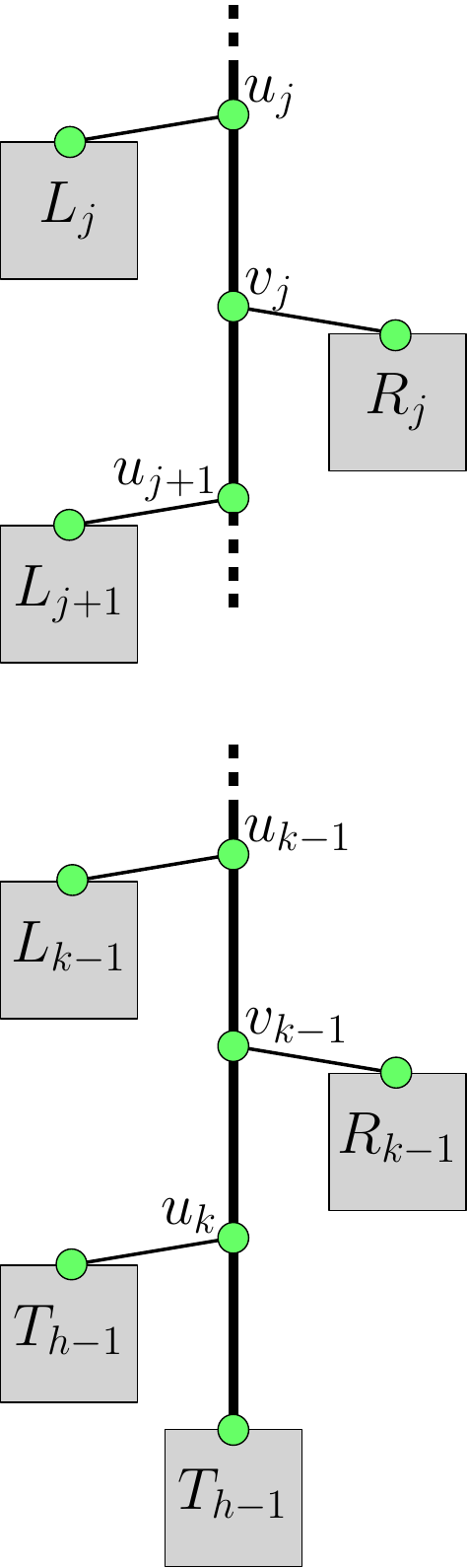}}
	\hfill
	\subfloat[]{
		\includegraphics[height=.45\textwidth,page=2]{Figures/lemma6.pdf}}
	\hfill \  
	\caption{Illustration for the proof of Lemma~\ref{le:LR-width}. (a) $\Gamma$ uses the left rule at $u_j,\dots,u_{k-1}$ and the right rule at $v_j,\dots,v_{k-1}$. (b)  $\Gamma$ uses the left rule at $u_j,\dots,u_{m-1}$ and the right rule at $v_j,\dots,v_{m-1},u_m$.}
	\label{fig:proof-lower}
\end{figure}

Suppose next (refer to Fig.~\ref{fig:proof-lower}(b)) that, for some integer $m$ with $j\leq m\leq k-1$, drawing $\Gamma$ is constructed by using the left rule at $u_j,\dots,u_{m-1}$, the right rule at $v_j,\dots,v_{m-1}$, and the right rule at $u_m$. Hence, nodes $u_j,\dots,u_m$ and $v_j,\dots,v_{m-1}$ are all aligned on the same vertical line $d$, however $v_m$ is to the right of $d$. Since $L_j,\dots,L_{m-1}$ lie to the left of $d$ in $\Gamma$, we have that $\alpha$ is larger than or equal to the widths of $L_j,\dots,L_{m-1}$. By the maximality of $j$, we have that $T(u_{m+1})$ requires width $2^{h-1}+k-(m+1)$ in any LR-drawing. Since the drawing of the subtree of $T_h$ rooted at $v_m$ is to the right of $d$ in $\Gamma$, it follows that the drawing of $T(u_{m+1})$ is also to the right of $d$ in $\Gamma$, hence $\beta\geq 2^{h-1}+k-(m+1)$. Now, if $m=j$, we have that $\alpha+\beta+1\geq (2^{h-1}+k-j-1) + 1 = 2^{h-1}+k-j$, a contradiction. Hence, we can assume that $m>j$. By induction, the width of the drawing of $L_i$ in $\Gamma$ is at least $\pi_{h-1}(i)$. Hence, the widths of the subtrees $L_j,\dots,L_{m-1}$ form a sequence of $m-j\geq 1$ consecutive elements of $\pi_{h-1}$. By Property~\ref{pr:sequence-max}, there exists an element $\pi_{h-1}(i)$ whose value is at least $m-j$. Then $\alpha\geq m-j$ and $\alpha+\beta+1\geq (m-j) +(2^{h-1}+k-m-1) + 1 = 2^{h-1}+k-j$, a contradiction.

Finally, the case in which, for some integer $m$ with $j\leq m\leq k-1$, drawing $\Gamma$ is constructed by using the left rule at $u_j,\dots,u_m$, the right rule at $v_j,\dots,v_{m-1}$, and the left rule at $v_m$ is symmetric to the previous one. This concludes the proof of the lemma. 
\end{proof}


\begin{lemma} \label{le:LR-nodes}
The number of nodes of $T_h$ is at most $(3+\sqrt 5)^h$.
\end{lemma} 

\begin{proof}
Denote by $n_h$ the number of nodes of tree $T_h$. By the way $T_h$ is recursively defined and since, for $i=0,\dots,h-2$, sequence $\sigma_{h-1}$ contains $2^{i}$ integers equal to $h-i-1$ (i.e., it contains one integer equal to $h-1$, two integers equal to $h-2$, $\dots$, $2^{h-2}$ integers equal to $1$), we have: 

\begin{eqnarray*}
n_h&=&\underbrace{(2n_{h-1}+1)}_{\textrm{subtree rooted at } u_k}+
	  \underbrace{(2(2^{h-1}-1))}_{\textrm{nodes } u_1,v_1,\dots,u_{k-1},v_{k-1}}
	+\underbrace{2(n_{h-1}+2n_{h-2}+\dots+2^{h-2}n_{1})}_{\textrm{subtrees of } u_1,v_1,\dots,u_{k-1},v_{k-1}}\\
      &=& 2n_{h-1}+ 2^{h}-1 + \sum_{i=1}^{h-1} 2^i n_{h-i}<2n_{h-1}+ 2^{h}+ \sum_{i=1}^{h-1} 2^i n_{h-i}.
\end{eqnarray*}

We now prove that $n_h\leq c^{h}$, for some constant $c$ to be determined later, by induction on $h$. The statement trivially holds for $h=1$, as long as $c\geq 1$, given that $n_1=1$. Now assume that $n_j\leq c^{j}$, for every $j\leq h-1$. Substituting $n_j\leq c^{j}$ into the upper bound for $n_h$ we get

\begin{eqnarray*}
n_h &\leq& 2c^{h-1}+2^{h}+ \sum_{i=1}^{h-1} 2^i c^{h-i}= 2 c^{h-1} + \sum_{i=1}^{h} 2^i c^{h-i}.
\end{eqnarray*}

By the factoring rule $c^{h+1}-2^{h+1} = (c-2)(c^h + 2c^{h-1}+\dots+2^{h-1}c + 2^h)$ we get  

\begin{eqnarray*}
\sum_{i=1}^{h} 2^i c^{h-i} &=&\frac{c^{h+1}-2^{h+1}}{c-2} - c^h = \frac{2c^{h}}{c-2}-\frac{2^{h+1}}{c-2}.
\end{eqnarray*}

Substituting that into the upper bound for $n_h$ we get
\begin{eqnarray*}
n_h &\leq& 2c^{h-1} + \frac{2c^{h}}{c-2}-\frac{2^{h+1}}{c-2} <2c^{h-1} + \frac{2c^{h}}{c-2}  =\frac{4c^{h}-4c^{h-1}}{c-2} ,
\end{eqnarray*}

where the second inequality holds as long as $c>2$.

Thus, we want $c$ to satisfy $\frac{4c^{h}-4c^{h-1}}{c-2}\leq c^h$; dividing by  $c^{h-1}$ and simplifying, the latter becomes $c^2-6c+4\geq 0$. The associated second degree equation has two solutions $c=3\pm \sqrt{5}$. Hence, $n_h\leq c^{h}$ holds true for $c\geq 3+\sqrt{5}$. This concludes the proof of the lemma. 
\end{proof}


Finally, we get the main result of this section.

\begin{theorem} \label{th:LR}
For infinitely many values of $n$, there exists an $n$-node ordered rooted binary tree that requires width $\Omega(n^{\delta})$ and area $\Omega(n^{1+\delta})$ in any LR-drawing, with $\delta=1/\log_2 (3+\sqrt 5)\geq 0.418$. 
\end{theorem}

\begin{proof}
By Lemma~\ref{le:LR-width} the width of any LR-drawing of $T_h$ is $w_h\geq 2^h-1$. Also, by Lemma~\ref{le:LR-nodes} tree $T_h$ has $n_h\leq (3+\sqrt 5)^h$ nodes, which taking the logarithms becomes $h\geq \log_{(3+\sqrt 5)} n_h$. Substituting this formula into the lower bound for the width, we get $w_h\geq 2^{\log_{(3+\sqrt 5)} n_h}-1$. Changing the base of the logarithm provides the statement about the width. Since any LR-drawing has height exactly $n$, the statement about the area follows.
\end{proof}


\subsection{Experimental Evaluation} \label{se:experiments}

It is tempting to evaluate $w^*_n$ by computing, for every $n$-node ordered rooted binary tree $T$, the minimum width $w^*_T$ of any LR-drawing of $T$ and by then taking the maximum among all such values. Although Theorem~\ref{thm:compute-matrix} ensures that $w^*_T$ can be computed efficiently, this evaluation is not practically possible, because of the large number of $n$-node ordered rooted binary trees, which is the $n$-th Catalan number ${{2n}\choose{n}}\frac{1}{n+1} \approx 4^n$; see, e.g.,~\cite{m-grbt-99}. 

We overcame this problem as follows. We say that a tree $T'$ {\em dominates} a tree $T$ if: (i) $n_{T'}\leq n_T$; (ii) $k_{T'}\geq k_{T}$; and (iii) for $i=0,\dots,k_{T}-1$, it holds ${\cal S}_{T'}(i) \geq {\cal S}_{T}(i)$. In order to perform an experimental evaluation of $w^*_n$, we construct a set ${\cal T}_n$ of ordered rooted binary trees with at most $n$ nodes such that every ordered rooted binary tree with at most $n$ nodes is dominated by a tree in ${\cal T}_n$. 

First, the dominance relationship ensures that, if an $n$-node ordered rooted binary tree exists requiring a certain width in any LR-drawing, then a tree in ${\cal T}_n$ also requires (at least) the same width in any LR-drawing (in a sense, the trees in ${\cal T}_n$ are the ``worst case'' trees for the width of an LR-drawing).     

Second, the size of ${\cal T}_n$ can be kept ``small'' by ensuring that no tree in ${\cal T}_n$ dominates another tree in ${\cal T}_n$. We could construct ${\cal T}_n$ for $n$ up to $455$, with ${\cal T}_{455}$ containing more than two million trees. 

Third, ${\cal T}_n$ can be constructed so that, for every $T\in {\cal T}_n$, the left and right subtrees of $r_T$ are also in ${\cal T}_n$. This is proved by induction on $|T|$. The base case $|T|=1$ is trivial. Further, if a tree $T$ in ${\cal T}_n$ has the left subtree $L$ of $r_T$ that is not in ${\cal T}_n$, then $L$ can be replaced with a tree in ${\cal T}_n$ that dominates $L$; this tree exists since $|L|<|T|$. This results in a tree $T'$ that dominates $T$. A similar replacement of the right subtree of $r_{T'}$ results in a tree $T''$ that dominates $T$ and such that the left and right subtrees of $r_{T''}$ are both in ${\cal T}_n$; then we replace $T$ with $T''$ in ${\cal T}_n$. Replacing all the trees with $|T|$ nodes in ${\cal T}_n$ completes the induction. Consequently, ${\cal T}_n$ can be constructed starting from ${\cal T}_{n-1}$ by considering a number of $n$-node trees whose size is quadratic in $|{\cal T}_{n-1}|$. Every time a tree $T$ is considered, its dominance relationship with every tree currently in ${\cal T}_n$ is tested. If a tree in ${\cal T}_n$ dominates $T$, then $T$ is discarded; otherwise, $T$ enters ${\cal T}_n$ and every tree in ${\cal T}_n$ that is dominated by $T$ is discarded. Note that the dominance relationship between two trees $T$ and $T'$ can be tested in time proportional to the size of ${\cal S}_{T}$ and ${\cal S}_{T'}$.



 





\newcolumntype{L}[1]{>{\raggedright\let\newline\\\arraybackslash\hspace{0pt}}m{#1}}
\newcolumntype{C}[1]{>{\centering\let\newline\\\arraybackslash\hspace{0pt}}m{#1}}
\newcolumntype{R}[1]{>{\raggedleft\let\newline\\\arraybackslash\hspace{0pt}}m{#1}}

By means of this approach, we were able to compute the value of $w^*_n$ for $n$ up to $455$. Table~\ref{ta:evaluation} shows the minimum integer $n$ such that there exists an $n$-node ordered rooted binary tree requiring a certain width $w$; for example, all the trees with up to $455$ nodes have LR-drawings with width at most $22$, and all the trees with up to $426$ nodes have LR-drawings with width at most $21$. Our experiments were performed with a monothread Java implementation on a machine with two $4$-core $3.16$GHz Intel(R) Xeon(R) CPU X$5460$ processors, with $48$GB of RAM, running Ubuntu $14.04.2$ LTS. The computation of the trees with $455$ nodes in ${\cal T}_{455}$ took more than one month.

\begin{table}[htb]
\centering
\begin{tabular} {|C{0.035\textwidth} |C{0.035\textwidth}|C{0.035\textwidth}|C{0.035\textwidth}|C{0.035\textwidth}|C{0.035\textwidth}|C{0.035\textwidth}|C{0.035\textwidth}|C{0.035\textwidth}|C{0.035\textwidth}|C{0.035\textwidth}|C{0.035\textwidth}|C{0.035\textwidth}|C{0.035\textwidth}|C{0.035\textwidth}|C{0.035\textwidth}|C{0.035\textwidth}|C{0.035\textwidth}|C{0.035\textwidth}|C{0.035\textwidth}|C{0.035\textwidth}|C{0.035\textwidth}|C{0.035\textwidth}|} 
\hline
{ $w$} & 1  & 2  & 3  & 4  & 5  & 6  & 7  & 8  & 9  & 10  & 11  & 12  & 13  & 14  & 15 & 16 & 17  & 18  & 19  & 20  & 21  & 22 \\
\hline
{ $n$} & 1 & 3 & 7  & 11 & 19  & 27  & 35  & 47  & 61  & 77  & 95  & 111  & 135  & 159  & 185  & 215  & 243 & 275 & 311 & 343 & 383 & 427\\
\hline
\end{tabular}
\vspace{1mm}
\caption{The table shows, for every integer $w$ between $1$ and $22$, the minimum number $n$ of nodes of a tree requiring $w$ width in any LR-drawing.}
\label{ta:evaluation}
\end{table}

We used the Mathematica software~\cite{mat} in order to find a function of the form $w = a \cdot n^b + c$ that better fits the values of Table~\ref{ta:evaluation}, according to the {\em least squares} optimization method (see, e.g.,~\cite{sti-gls-81}). Recall that by Theorem~\ref{th:LR} and by Chan results~\cite{c-nlabdbt-02}, $w^*_n$ is asymptotically between $\Omega(n^{0.418})$ and $O(n^{0.48})$. We obtained $w = 1.54002 \cdot n^{0.443216} -0.549577$ as an optimal function; see Fig.~\ref{fig:interpolation}. This seems to indicate that the best known upper and lower bounds are not tight. 

\begin{figure}[htb]
	\begin{center}
		\includegraphics[width=.3\textwidth]{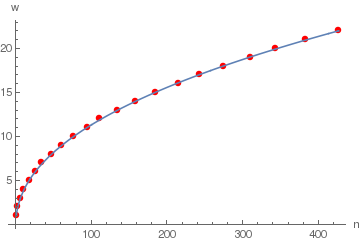} 
		\caption{Function $w = 1.54002 \cdot n^{0.443216} -0.549577$ (blue line) and data from Table~\ref{ta:evaluation} (red dots).}
		\label{fig:interpolation}
	\end{center}
\end{figure}

  \begin{figure}[htb]
    \centering
    \subfloat[]{
		\includegraphics[width=.35\textwidth]{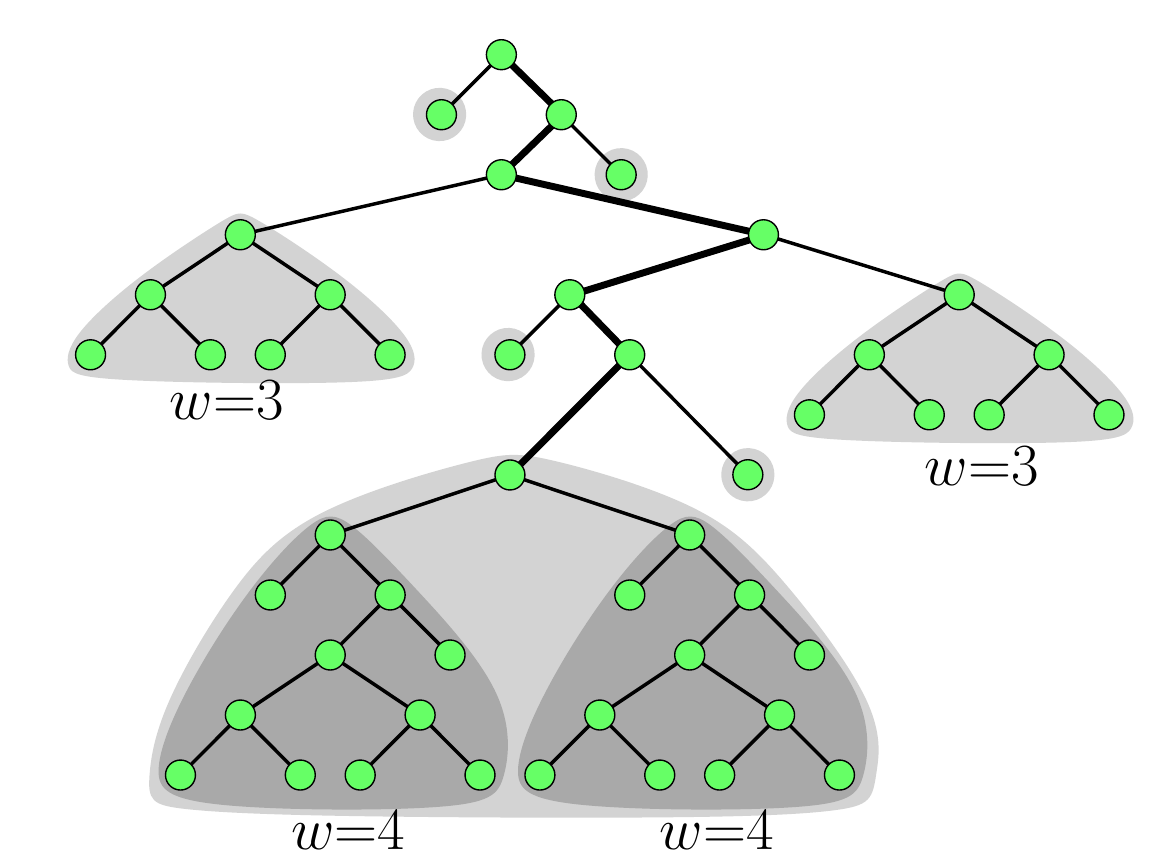}}
    \subfloat[]{
		\includegraphics[width=.35\textwidth]{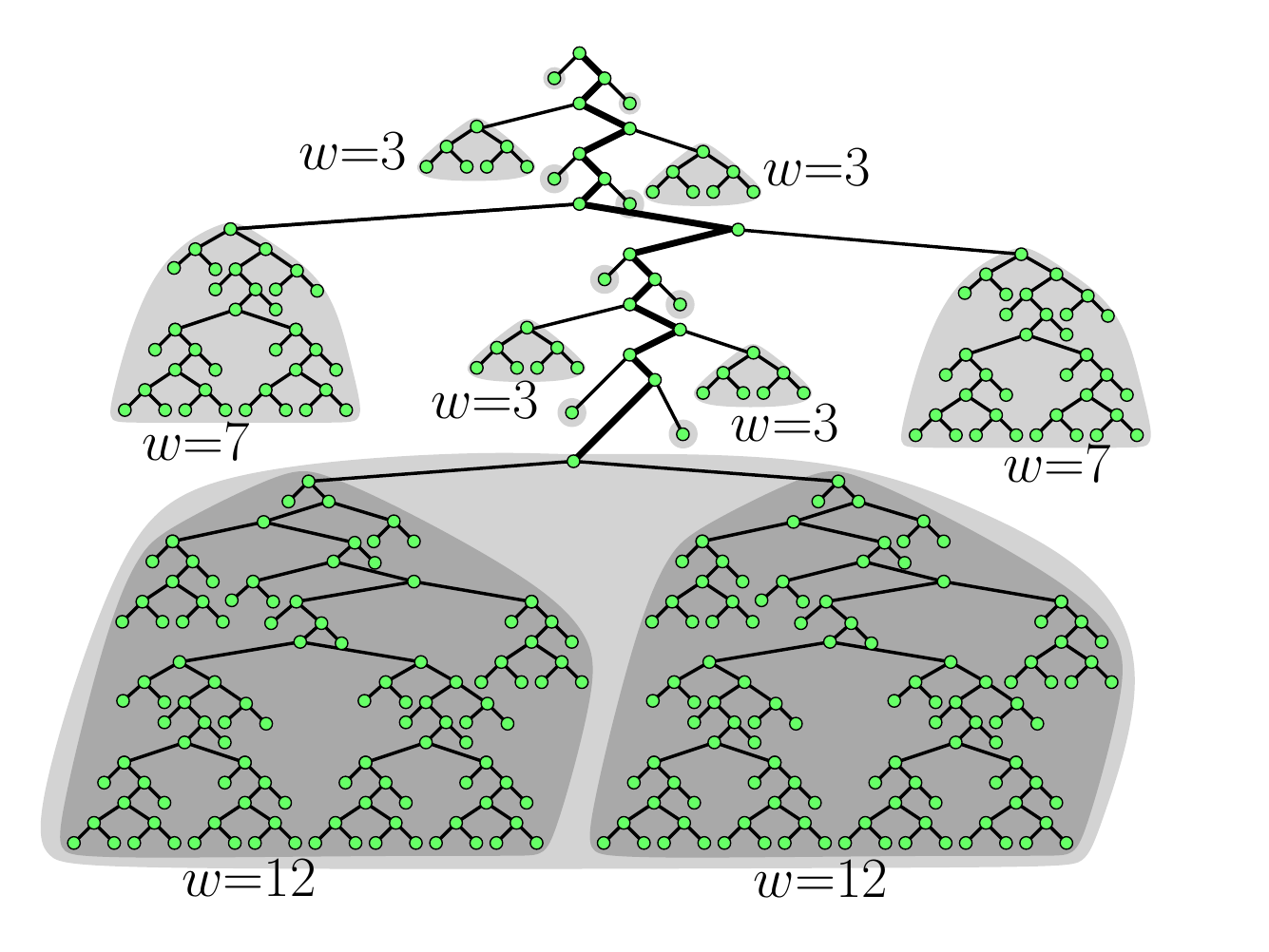}}
    \caption{(a) A tree with $n=47$ nodes requiring width $8$ in any LR-drawing. (b) A tree with $n=343$ nodes requiring width $20$ in any LR-drawing.}
    \label{fig:lower-trees}
  \end{figure}

As a final remark, we note that the structure of the trees corresponding to the pairs $(n,w)$ in Table~\ref{ta:evaluation} (see Fig.~\ref{fig:lower-trees}) is similar to the structure of the trees that provide the lower bound of Theorem~\ref{th:LR}, which might indicate that the lower bound is close to be tight: In particular, the left (and right) subtrees of the thick path in Fig.~\ref{fig:lower-trees}(b) require width $1,3,1,7,1,3,1$ from top to bottom,  as in the lower bound tree $T_4$ from Theorem~\ref{th:LR}; also, the subtrees of the last node of the thick path are isomorphic, as in $T_4$ (although these subtrees require width $7$ in $T_4$, while they require width $12$ in Fig.~\ref{fig:lower-trees}(b)).



\section{Straight-Line Drawings of Outerplanar Graphs} \label{se:outerplanar}

In this section we study outerplanar straight-line drawings of outerplanar graphs. 

\subsection{From Outerplanar Drawings to Star-Shaped Drawings} \label{se:outerplanar-star-shaped}

Let $G$ be a maximal outerplanar graph, that is, a graph to which no edge can be added without violating its outerplanarity. We assume that $G$ is associated with any (not necessarily straight-line) outerplanar drawing. This allows us to talk about the faces of $G$, rather than about the faces of a drawing of $G$. We denote by $f^*$ the outer face of $G$. The {\em dual tree} $T$ of $G$ has a node for each face $f\neq f^*$ of $G$ (we denote by $f$ both the face of $G$ and the corresponding node of $T$); further, $T$ has an edge $(f_1,f_2)$ if the faces $f_1$ and $f_2$ of $G$ share an edge $e$ along their boundaries; we say that $e$ and $(f_1,f_2)$ are {\em dual} to each other. 

  \begin{figure}[htb]
    \centering
    \hfill
	\subfloat[]{
		\includegraphics[page=1,width=.4\textwidth]{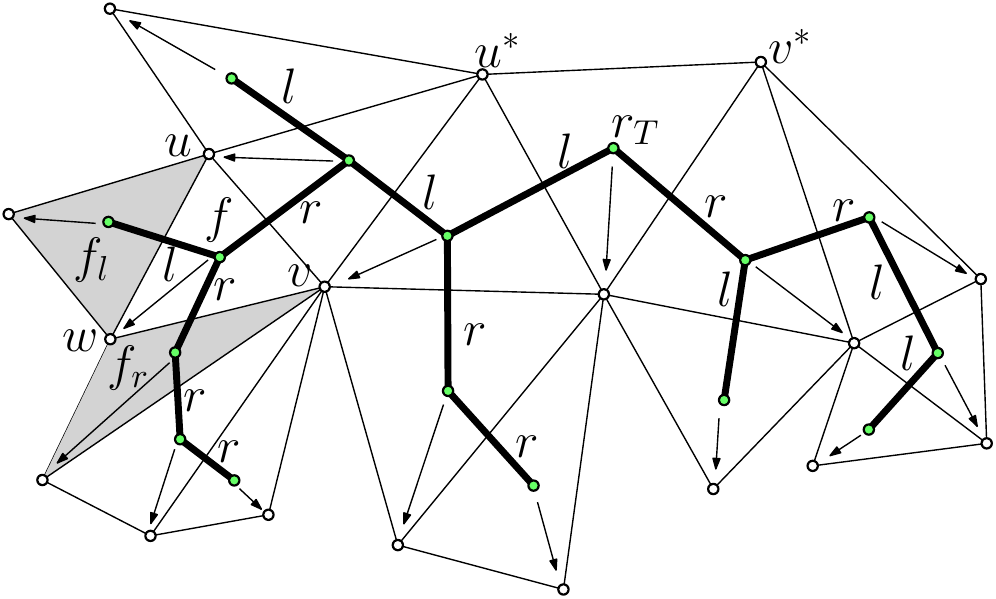}}
    \hfill
	\subfloat[]{
		\includegraphics[page=2,width=.4\textwidth]{Figures/outerplanar.pdf}}
    \hfill \
	\caption{(a) A maximal outerplanar graph $G$ (shown with white circles and thin line segments) and its dual tree $T$ (shown with green circles and thick line segments). The labels $l$ and $r$ on the edges of $T$ show whether a node is the left or the right child of its parent, respectively. The gray faces $f_l$ and $f_r$ are the left and the right child of the face $f$. The arrows show a bijective mapping $\gamma$ from the nodes of $T$ to the vertices of $G'$ such that an edge $(s,t)$ belongs to $T$ if and only if the edge $(\gamma(s),\gamma(t))$ belongs to $G'$. (b) A star-shaped drawing $\Gamma_T$ of $T$ (shown with green circles and thick line segments). The gray regions show the polygons $P_l(s)$ and $P_r(t)$ for two nodes $s$ and $t$ of $T$. Adding the thin edges and the white vertices at $p^*_u$ and $p^*_v$ turns $\Gamma_T$ into an outerplanar straight-line drawing $\Gamma_G$ of $G$.}
    \label{fig:outerplanar-starshaped}
  \end{figure}
  
We now turn $T$ into an ordered rooted binary tree. Refer to Fig.~\ref{fig:outerplanar-starshaped}(a). First, pick any edge $(u^*,v^*)$ incident to $f^*$, where $v^*$ is encountered right after $u^*$ when walking in clockwise direction along the boundary of $f^*$; root $T$ at the node corresponding to the internal face of $G$ incident to $(u^*,v^*)$. Second, since $G$ is maximal, all its internal faces are delimited by cycles with $3$ vertices, hence $T$ is binary. Third, an outerplanar drawing of $G$ naturally defines whether a child of a node of $T$ is a left or right child. Namely, consider any non-leaf node $f$ of $T$. If $f\neq f^*$, then let $g$ be the parent of $f$ and let $(u,v)$ be the edge of $G$ dual to $(f,g)$. If $f=f^*$, then let $u=u^*$ and $v=v^*$. In both cases, let $w\neq u,v$ be the third vertex of $G$ incident to $f$; assume, w.l.o.g. that $u$, $v$, and $w$ appear in this clockwise order along the boundary of $f$. Let $(f,f_l)$ and $(f,f_r)$ be the edges of $T$ dual to $(u,w)$ and $(v,w)$, respectively. Then $f_l$ and $f_r$ are the left and right child of $f$, respectively; note that one of these children might not exist (if $(u,w)$ or $(v,w)$ is incident to $f^*$). Henceforth, we regard $T$ as an ordered rooted binary tree.

We introduce some definitions. The {\em leftmost} ({\em rightmost}) {\em path} of $T$ is the maximal path $s_0,\dots,s_m$ such that $s_0=r_T$ and $s_{i}$ is the left (resp.\ right) child of $s_{i-1}$, for $i=1,\dots,m$. For a node $s$ of $T$, the {\em left-right} ({\em right-left}) {\em path} of $s$ is the maximal path $s_0,\dots,s_m$ such that $s_0=s$, $s_1$ is the left (resp.\ right) child of $s_0$, and $s_{i}$ is the right (resp.\ left) child of $s_{i-1}$, for $i=2,\dots,m$. For a node $s$ of $T$, let $C_l(s)$ (resp.\ $C_r(s)$) denote the cycle composed of the left-right (resp.\ right-left) path $s_0,\dots,s_m$ of $s$ plus edge $(s_0,s_m)$ -- this cycle degenerates into a vertex or an edge if $m=0$ or $m=1$, respectively. Finally, a drawing of $T$ is {\em star-shaped} if it satisfies the following properties (refer to Fig.~\ref{fig:outerplanar-starshaped}(b)):

\begin{enumerate}
\item The drawing is planar, straight-line, and {\em order-preserving} (that is, for every degree-$3$ node $s$ of $T$, the edge between $s$ and its parent, the edge between $s$ and its left child, and the edge between $s$ and its right child appear in this counter-clockwise order around $s$). 
\item For each node $s$ of $T$, draw the edge of $C_l(s)$ not in $T$ (if such an edge exists) as a straight-line segment and let $P_l(s)$ be the polygon representing $C_l(s)$. Then $P_l(s)$ is simple (that is, not self-intersecting) and every straight-line segment between $s$ and a non-adjacent vertex of $P_l(s)$ lies inside $P_l(s)$. A similar condition is required for the polygon $P_r(s)$ representing $C_r(s)$.
\item For any node $s$ of $T$, the polygons $P_l(s)$ and $P_r(s)$ lie one outside the other, except at $s$; also, for any two distinct nodes $s$ and $t$ of $T$, the polygons $P_l(s)$ and $P_r(s)$ lie outside polygons $P_l(t)$ and $P_r(t)$, and vice versa, except at common vertices and edges along their boundaries. 
\item There exist two points $p^*_{u}$ and $p^*_{v}$ such that the straight-line segments connecting $p^*_{u}$ with the nodes of the leftmost path of $T$, connecting $p^*_{v}$ with the nodes of the rightmost path of $T$, and connecting $p^*_{u}$ with $p^*_{v}$ do not intersect each other and, for any node $s$ of $T$, they lie outside polygons $P_l(s)$ and $P_r(s)$, except at common vertices.
\end{enumerate}

We now describe the key ideas developed in~\cite{bf-sadog-09} in order to relate outerplanar straight-line drawings of outerplanar graphs to star-shaped drawings of their dual trees. Let $G$ be a maximal outerplanar graph and $T$ be its dual tree; also, let $G'$ be the graph obtained from $G$ by removing vertices $u^*$ and $v^*$ and their incident edges. Then $T$ is a subgraph of $G'$; in fact, there exists a bijective mapping $\gamma$ from the nodes of $T$ to the vertices of $G'$ such that an edge $(s,t)$ belongs to $T$ if and only if the edge $(\gamma(s),\gamma(t))$ belongs to $G'$ (see Fig.~\ref{fig:outerplanar-starshaped}(a)). Further, the graph obtained by adding to $T$, for every node $s$ in $T$, edges connecting $s$ with all the (not already adjacent) nodes on the left-right and on the right-left path of $s$ is $G'$. Properties 1--3 of a star-shaped drawing ensure that, in order to obtain an outerplanar straight-line drawing of $G'$, one can start from a star-shaped drawing of $T$ and just draw the edges of $G'$ not in $T$ as straight-line segments. Finally, an outerplanar straight-line drawing of $G$ is obtained by mapping $u^*$ and $v^*$ to $p^*_{u}$ and $p^*_{v}$ (defined as in Property 4 of a star-shaped drawing), respectively, and by drawing their incident edges as straight-line segments (see Fig.~\ref{fig:outerplanar-starshaped}(b)). 

If one starts from a star-shaped drawing $\Gamma_T$ of $T$ in a certain area $A$, an outerplanar straight-line drawing $\Gamma_G$ of $G$ can be constructed as described above; then the area of $\Gamma_G$ might be larger than $A$, since points $p^*_{u}$ and $p^*_{v}$ might lie outside the bounding box of $\Gamma_T$. However, $\Gamma_G$ is equal to the area of the smallest axis-parallel rectangle\footnote{By the {\em width} and the {\em height} of a rectangle we mean the number of grid columns and rows intersecting it, respectively. By the {\em area} of a rectangle we mean its width times its height.} containing $p^*_{u}$, $p^*_{v}$, and $\Gamma_T$. We formalize this in the following.

\begin{lemma} \label{le:star-shaped-correspondence} (Di Battista and Frati~\cite{bf-sadog-09}) 
If $T$ admits a star-shaped drawing $\Gamma_T$, then $G$ admits an outerplanar straight-line drawing $\Gamma_G$ whose area is equal to the area of the smallest axis-parallel rectangle containing $p^*_{u}$, $p^*_{v}$, and $\Gamma_T$.
\end{lemma}

In the next sections we will show algorithms for constructing star-shaped drawings $\Gamma_T$ of ordered rooted binary trees $T$ in which the smallest axis-parallel rectangle containing $\Gamma_T$, $p^*_{u}$, and $p^*_{v}$ has asymptotically the same area as $\Gamma_T$. 

\subsection{Star-Shaped Drawings with $O(\omega)$ Width} \label{se:weak-star-shaped}

In this section we show that, if an ordered rooted binary tree admits an LR-drawing with width $\omega$, then it admits a star-shaped drawing with width $O(\omega)$. In fact, we will prove the existence of two star-shaped drawings with that width, each satisfying some additional geometric properties. Because of the similarity of our constructions with the ones in~\cite{bf-sadog-09}, we will not prove formally that the constructed drawings are star-shaped, and we will only provide the main intuition for that. Further, the illustrations of our constructions will show the points $p^*_u$ and $p^*_v$ (represented by white disks) and the straight-line segments (represented by gray lines) to be added to the star-shaped drawings according to Properties~2 and~4 from Section~\ref{se:outerplanar-star-shaped}. Given a drawing $\Gamma$ of a tree, we often say that a vertex $u$ {\em sees} another vertex $v$ if the straight-line segment between $u$ and $v$ does not cross $\Gamma$.

Consider a star-shaped drawing $\Gamma$ of an ordered rooted binary tree $T$. Denote by $B_l(\Gamma)$, $B_t(\Gamma)$, $B_r(\Gamma)$, and $B_b(\Gamma)$ the left, top, right, and bottom side of $B(\Gamma)$, respectively. 

  \begin{figure}[htb]
  	\centering
  	\hfill
  	\subfloat[]{
  		\includegraphics[scale=0.9]{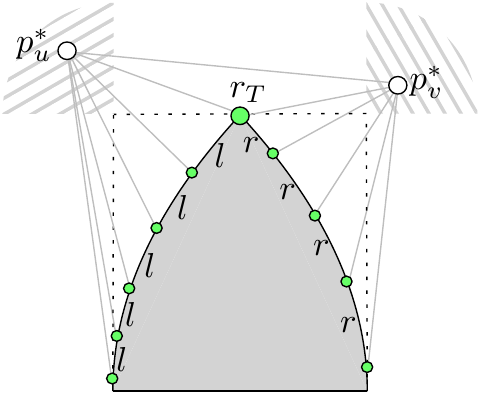}}
  	\hfill
  	\subfloat[]{
  		\includegraphics[scale=0.9]{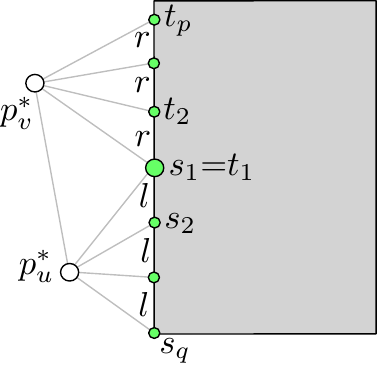}}
  	\hfill \
  	\caption{(a) A schematization of the shape of a bell-like star-shaped drawing. (b) A schematization of the shape of a flat star-shaped drawing.}
  	\label{fig:shapes}
  \end{figure}

We say that $\Gamma$ is {\em bell-like} (see Fig.~\ref{fig:shapes}(a)) if: (i) $r_T$ lies on $B_t(\Gamma)$; and (ii) any point $p^*_u$ above $B_t(\Gamma)$ and to the left of $B_l(\Gamma)$ and any point $p^*_v$ above $B_t(\Gamma)$ and to the right of $B_r(\Gamma)$ satisfy Property 4 of a star-shaped drawing. 

We say that $\Gamma$ is {\em flat} (see Fig.~\ref{fig:shapes}(b)) if: (i) the leftmost path $(s_1=r_T,\dots,s_q)$ and the rightmost path $(t_1=r_T,\dots,t_p)$ of $T$ lie on $B_l(\Gamma)$; and (ii) $y(s_{i-1})>y(s_i)$, for $i=2,\dots,q$, and $y(t_{i-1})<y(t_i)$, for $i=2,\dots,p$. 

We now present the main lemma of this section.

\begin{lemma}\label{le:weak-star-shaped}
Consider an $n$-node ordered rooted binary tree $T$ and suppose that $T$ admits an LR-drawing with width $\omega$. Then $T$ admits a bell-like star-shaped drawing with width at most $4\omega-2$ and height at most $n$, and a flat star-shaped drawing with width at most $4\omega$ and height at most $n$. 
\end{lemma}

In the remainder of the section we prove Lemma~\ref{le:weak-star-shaped} by exhibiting two algorithms, called {\em bell-like algorithm} and {\em flat algorithm}, that construct bell-like and flat star-shaped drawings of trees, respectively. Both algorithms use induction on $\omega$; each of them is defined in terms of the other one. The base case of both algorithms is $\omega=1$. This implies that $T$ is a root-to-leaf path $(v_1=r_T,\dots,v_n)$, as in Fig.~\ref{fig:starshaped-weak-base}(a).

  \begin{figure}[htb]
    \centering
    \hfill
	\subfloat[]{
		\includegraphics[page=1,width=.13\textwidth]{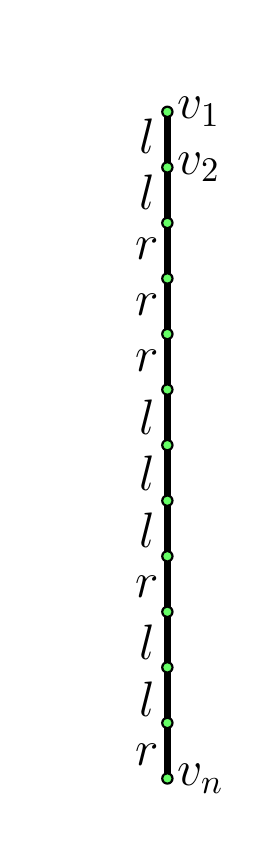}}
    \hfill
	\subfloat[]{
		\includegraphics[page=2,width=.13\textwidth]{Figures/basecases_outerplanar.pdf}}
    \hfill
	\subfloat[]{
		\includegraphics[page=1,width=.18\textwidth]{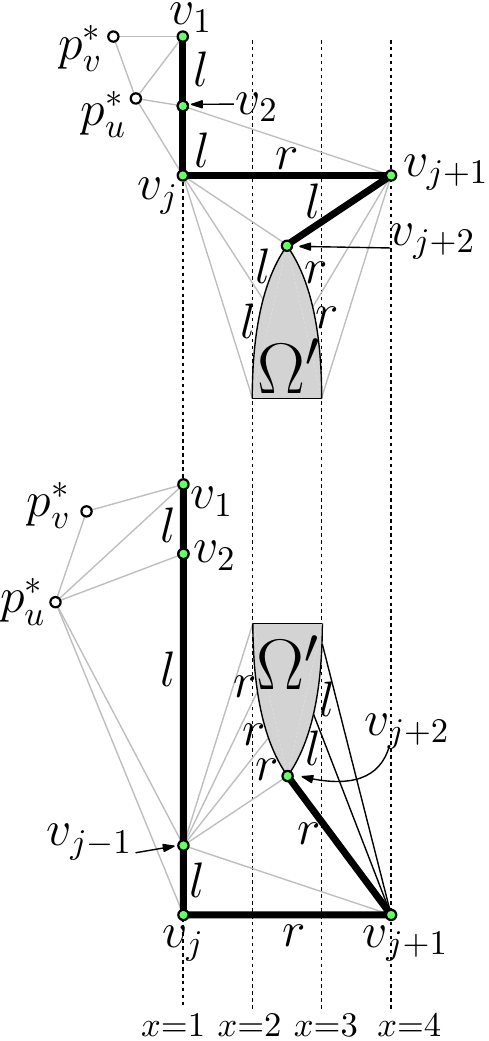}}
    \hfill \
	\caption{(a) An LR-drawing with width $1$ of a tree $T$. (b) A bell-like star-shaped drawing with width $2$ of $T$. Any points $p^*_u$ and $p^*_v$ in the shaded regions see all the nodes of the leftmost and rightmost path of $T$, respectively. (c) A flat star-shaped drawing with width $4$ of $T$, if $v_{j+2}$ is the left (see top) or right (see bottom) child of $v_{j+1}$.}
    \label{fig:starshaped-weak-base}
  \end{figure}
  
The bell-like algorithm constructs a bell-like star-shaped drawing $\Omega$ of $T$ as follows (refer to Fig.~\ref{fig:starshaped-weak-base}(b)). For $i=1,\dots,m$, set $y(v_i)=-i$. Also, set $x(v_i)=2$ for every node $v_i$ such that $i<n$ and such that $v_{i+1}$ is the left child of $v_i$, and set $x(v_i)=1$ for every other node $v_i$. Then $\Omega$ has width at most $2=4\omega-2$ and height $n$. Further, $\Omega$ is readily seen to be a bell-like star-shaped drawing. In particular, the left-right path of each node $v_i$ is either a single node, or is a single edge, or is represented by the legs and the base with smaller length of an isosceles trapezoid (which possibly degenerates to a triangle); thus $v_i$ sees its left-right path. Similarly $v_i$ sees its right-left path, and hence $\Omega$ satisfies Property 2 of a star-shaped drawing. Moreover, the leftmost path of $T$ is either a single node (if $v_2$ is the right child of $v_1$) or is a polygonal line that is strictly decreasing in the $y$-direction and non-increasing in the $x$-direction from $r_T$ to its last node. A similar argument for the rightmost path, together with the fact that $r_T$ lies on $\Omega_t$, implies that $\Omega$ satisfies the bell-like property.


The flat algorithm constructs a flat star-shaped drawing $\Pi$ of $T$ as follows (refer to Fig.~\ref{fig:starshaped-weak-base}(c)). Assume that $v_2$ is the left child of $v_1$; the other case is symmetric. Let $(v_1,\dots,v_j)$ be the leftmost path of $T$, where $j\geq 2$. If $j=n$, then $\Pi$ is constructed by setting $x(v_i)=1$ and $y(v_i)=-i$, for $i=1,\dots,n$ (then $\Pi$ has width $1<4\omega$ and height $n$). Otherwise, $v_{j+1}$ is the right child of $v_j$. Use the bell-like algorithm to construct a bell-like star-shaped drawing $\Omega'$ with width at most $2$ of the subtree of $T$ rooted at $v_{j+2}$ (note that this subtree has an LR-drawing with width $1$ since $T$ does). We distinguish two cases.

\begin{itemize}
\item If $v_{j+2}$ is the left child of $v_{j+1}$ (as in Fig.~\ref{fig:starshaped-weak-base}(c) top), then set $x(v_i)=1$ and $y(v_i)=-i$, for $i=1,\dots,j$, $x(v_{j+1})=4$, and $y(v_{j+1})=-j$. Place $\Omega'$ so that $B_t(\Omega')$ is on the line $y=-j-1$, and so that $B_l(\Omega')$ is on the line $x=2$. Since $v_j$ is above $B_t(\Omega')$ and to the left of $B_l(\Omega')$, it sees all the nodes of its right-left path, given that $\Omega'$ satisfies the bell-like property; since $v_{j+1}$ is above $B_t(\Omega')$ and to the right of $B_r(\Omega')$, it sees all the nodes of its left-right path; hence $\Pi$ satisfies Property 2 of a star-shaped drawing.
\item If $v_{j+2}$ is the right child of $v_{j+1}$ (as in Fig.~\ref{fig:starshaped-weak-base}(c) bottom), then set $x(v_j)=1$, $y(v_j)=0$, $x(v_{j-1})=1$, and $y(v_i)=1$; rotate $\Omega'$ by $180^{\circ}$ and place it so that $B_b(\Omega')$ is on the line $y=2$, and so that $B_l(\Omega')$ is on the line $x=2$; finally, place vertices $v_{1},\dots,v_{j-2}$, if any, on the line $x=1$, so that $v_{j-2}$ is one unit above $B_t(\Omega')$, and so that $y(v_i)=y(v_{i+1})+1$, for $i=1,\dots,j-3$. Since $v_{j-1}$ is below $B_b(\Omega')$ and to the left of $B_l(\Omega')$, it sees all the nodes of its left-right path, given that $\Omega'$ is rotated by $180^{\circ}$ and satisfies the bell-like property; since $v_{j+1}$ is below $B_b(\Omega')$ and to the right of $B_r(\Omega')$, it sees all the nodes of its right-left path; hence $\Pi$ satisfies Property 2 of a star-shaped drawing.
\end{itemize}

In both cases the leftmost path of $T$ lies on $B_l(\Pi)$, with $r_T=v_1$ as the vertex with largest $y$-coordinate; hence $\Pi$ satisfies the flat property. This concludes the description of the base case.

We now discuss the inductive case, in which $\omega>1$. Refer to Fig.~\ref{fig:starshaped-weak-inductive}(a). Let $\Gamma$ be an LR-drawing of $T$ with width $\omega$; let $\omega_l$ and $\omega_r$ be the left and right width of $\Gamma$, respectively; we are going to use $\omega_l+\omega_r+1=\omega$, which holds by Property~\ref{pr:width} of an LR-drawing; in particular, $\omega_l,\omega_r<\omega$. Define a path $P=(v_1,\dots,v_m)$ as follows. First, let $v_1=r_T$; for $i=1,\dots,m-1$, node $v_{i+1}$ is the left or right child of $v_i$, depending on whether $\Gamma$ uses the right or the left rule at $v_i$, respectively; finally, $v_m$ is either a leaf, or a node with no left child at which $\Gamma$ uses the right rule, or a node with no right child at which $\Gamma$ uses the left rule. Note that $P$ lies on a single vertical line in $\Gamma$. Denote by $l_i$ or $r_i$ the child not in $P$ of $v_i$, depending on whether that node is a left or right child of $v_i$, respectively; denote by $L_i$ (by $R_i$) the subtree of $T$ rooted at $l_i$ (resp.\ $r_i$). Note that $L_i$ ($R_i$) admits an LR-drawing with width at most $\omega_l$ (resp.\ $\omega_r$), hence by induction it also admits a bell-like star-shaped drawing with width at most $4\omega_l-2$ (resp.\ $4\omega_r-2$), and a flat star-shaped drawing with width at most $4\omega_l$ (resp.\ $4\omega_r$).

  \begin{figure}[htb]
    \centering
    \hfill
	\subfloat[]{
		\includegraphics[page=2,height=.5\textwidth]{Figures/bells-rect.pdf}}
    \hfill
	\subfloat[]{
		\includegraphics[page=1,height=.5\textwidth]{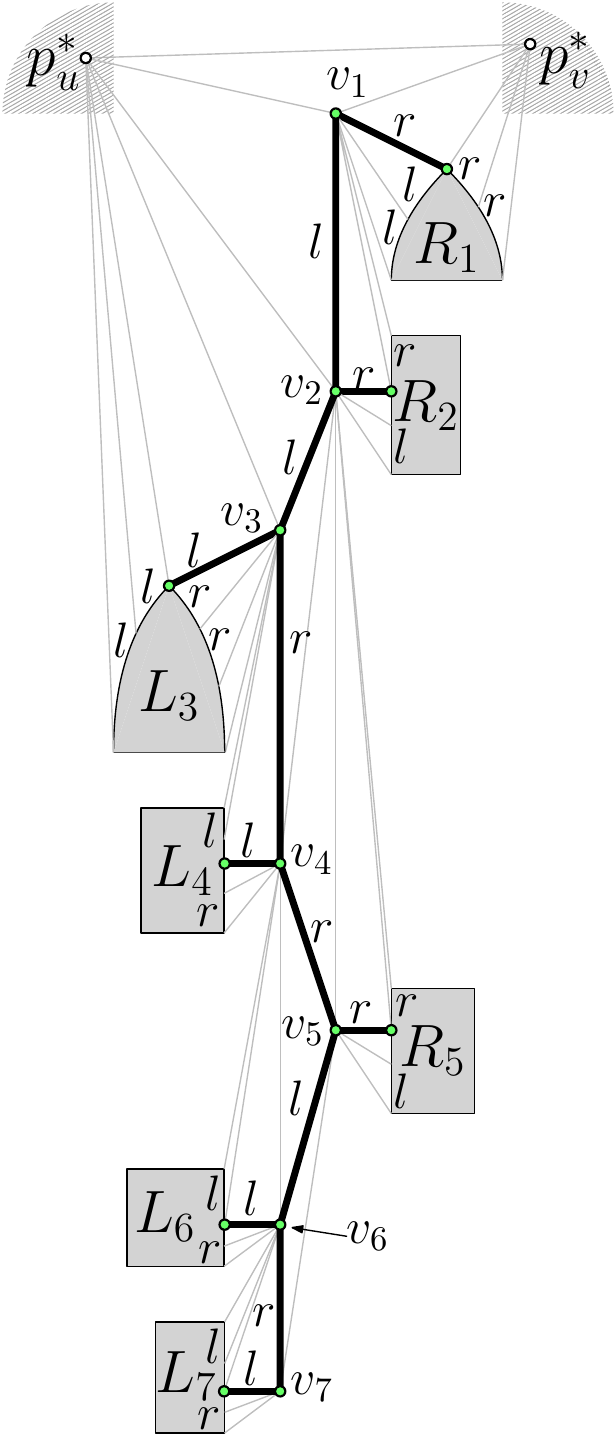}}
    \hfill
	\subfloat[]{
		\includegraphics[page=3,height=.5\textwidth]{Figures/bells-rect.pdf}}
    \hfill \
	\caption{(a) An LR-drawing $\Gamma$ of $T$. (b) A bell-like star-shaped drawing $\Omega$ of $T$. In this example $j=3$ and $h=1$. (c) A flat star-shaped drawing $\Pi$ of $T$. In this example $j=3$, $p=4$, and $q=6$.}
    \label{fig:starshaped-weak-inductive}
  \end{figure}
  
The bell-like algorithm constructs a bell-like star-shaped drawing $\Omega$ of $T$ as follows. Refer to Fig.~\ref{fig:starshaped-weak-inductive}(b).  Let $j\geq 1$ ($h\geq 1$) be the smallest index such that $\Gamma$ uses the left (resp.\ right) rule at $v_j$. Index $j$ ($h$) might be undefined if $\Gamma$ uses the right (resp.\ left) rule at every node of $P$. Inductively construct a bell-like star-shaped drawing $\Omega_j$ of $L_j$ (if this subtree exists) and a bell-like star-shaped drawing $\Omega_h$ of $R_h$ (if this subtree exists); inductively construct a flat star-shaped drawing $\Pi_i$ of every other subtree $L_i$ or $R_i$ of $P$. Similarly to the base case, set $x(v_i)=2$ or $x(v_i)=1$, depending on whether the left child of $v_i$ is $v_{i+1}$ or not, respectively. Next, we define the placement of $\Omega_j$, of $\Omega_h$, and of each $\Pi_i$ with respect to  $v_j$, $v_h$, and $v_i$, respectively. Drawing $\Omega_j$ ($\Omega_h$) is placed so that $B_r(\Omega_j)$ ($B_l(\Omega_h)$) lies on the line $x=0$ (resp.\ $x=3$) and so that $B_t(\Omega_j)$ ($B_t(\Omega_h)$) is one unit below $v_j$ (resp.\ $v_h$). For every right subtree $R_i\neq R_h$ of $P$, drawing $\Pi_i$ is placed so that $B_l(\Pi_i)$ lies on the line $x=3$ and so that $y(r_i)=y(v_i)$; further, for every left subtree $L_i\neq L_j$ of $P$, drawing $\Pi_i$ is first rotated by $180^\circ$, and then it is placed so that $B_r(\Pi_i)$ lies on the line $x=0$ and so that $y(l_i)=y(v_i)$. Finally, for $i=1,\dots,m-1$, set $y(v_i)$ so that the bottom side of the smallest axis-parallel rectangle containing $v_i$ and the drawing of its subtree $L_i$ or $R_i$ is one unit above the top side of the smallest axis-parallel rectangle containing $v_{i+1}$ and the drawing of its subtree $L_{i+1}$ or $R_{i+1}$. This completes the construction of $\Omega$. The height of $\Omega$ is at most $n$, since every grid row intersecting $\Omega$ contains a node of $P$ or intersects a subtree of $P$. Further, the width of $\Omega$ is equal to the maximum width of the drawing of a subtree $L_i$, which is at most $4\omega_l$ by induction, plus the maximum width of the drawing of a subtree $R_i$, which is at most $4\omega_r$ by induction, plus two, since the nodes of $P$ lie on two grid columns. Hence the width of $\Omega$ is at most $4\omega_l+4\omega_r+2=4\omega-2$. The leftmost path of $T$ is composed of the path $(v_1,\dots,v_j,l_j)$ and of the leftmost path of $L_j$. Since $(v_1,\dots,v_j,l_j)$ is represented in $\Omega$ by a polygonal line that is strictly decreasing in the $y$-direction and non-increasing in the $x$-direction from $v_1$ to $l_j$, and since every point to the left of $B_l(\Omega_j)$ and above $B_t(\Omega_j)$ sees all the nodes of the leftmost path of $L_j$, by induction, we get that every point to the left of $B_l(\Omega)$ and above $B_t(\Omega)$ sees all the nodes of the leftmost path of $T$. A similar argument for the rightmost path, together with the fact that $r_T$ lies on $B_t(\Omega)$, implies that $\Omega$ satisfies the bell-like property. Concerning Property 2 of a star-shaped drawing, we note that $v_j$ sees all the nodes of its left-right path since it is above $B_t(\Omega_j)$ and to the right of $B_r(\Omega_j)$, and since $\Omega_j$ satisfies the bell-like property. Also, if $v_{i+1}$ is the left child of $v_i$ and $v_{i+2}$ is the right child of $v_{i+1}$, as with $i=2$ in Fig.~\ref{fig:starshaped-weak-inductive}(b), then the representation of the left-right path of $v_i$ in $\Omega$ consists of the legs and of the base with smaller length of a trapezoid, of a horizontal segment between the lines $x=2$ and $x=3$, and of a vertical segment on the line $x=3$; hence $v_i$ sees all the nodes of its left-right path. 

The flat algorithm constructs a flat star-shaped drawing $\Pi$ of $T$ as follows. Refer to Fig.~\ref{fig:starshaped-weak-inductive}(c). Assume that $v_2$ is the left child of $v_1$; the other case is symmetric. 

First, we construct a drawing $\Pi_R$ of the right subtree $R_1$ of $r_T$. Let $(t_1=r_T,\dots,t_p)$ be the rightmost path of $T$. For $i=2,\dots,p$, let $T_l(i)$ be the left subtree of $t_i$. Since $v_2$ is the left child of $v_1$, drawing $\Gamma$ uses the right rule at $v_1$, hence $R_1$ admits an LR-drawing with width at most $\omega_r$. Tree $T_l(i)$ also admits an LR-drawing with width at most $\omega_r$, given that it is a subtree of $R_1$. By induction $T_l(i)$ admits a flat star-shaped drawing $\Pi_l(i)$ with width at most $4\omega_r\leq 4\omega-4$. Set $x(t_i)=1$ for $i=2,\dots,p$. Next, we define the placement of each $\Pi_l(i)$ with respect to $t_i$. Drawing $\Pi_l(i)$ is placed so that $B_l(\Pi_l(i))$ is on the line $x=2$ and so that the root of $T_l(i)$ is on the same horizontal line as $t_i$.  Finally,  set $y(t_i)$ so that, for $i=3,\dots,p$, the bottom side of the smallest axis-parallel rectangle containing $t_i$ and $\Pi_l(i)$ is one unit above the top side of the smallest axis-parallel rectangle containing $t_{i-1}$ and $\Pi_l(i-1)$. This completes the construction of $\Pi_R$.  

Second, we construct a drawing $\Pi_L$ of the left subtree $L_1$ of $r_T$. Let $(s_1=r_T,\dots,s_q)$ be the leftmost path of $T$ and, for $i=2,\dots,q$, let $T_r(i)$ be the right subtree of $s_i$. Further, let $j\geq 2$ be the largest integer for which $(s_1,\dots,s_j)$ belongs to $P$; that is, $s_i=v_i$ holds true for $i=1,\dots,j$. Although $v_j$ might be the last node of $P$, we assume that $v_{j+1}$ exists; the construction for the case in which $v_{j+1}$ does not exist is much simpler. By the maximality of $j$, we have that $v_{j+1}$ is the right child of $v_j$. Let $C$ and $D$ be the left and right subtrees of $v_{j+1}$, respectively (possibly one or both of these subtrees are empty). Each of $C$ and $D$ admits an LR-drawing with width $\omega$, given that $\Gamma$ has width $\omega$. Construct bell-like star-shaped drawings $\Omega_C$ of $C$ and $\Omega_D$ of $D$ with width at most $4\omega-2$. Note that, for $i=2,\dots,j-1$, drawing $\Gamma$ uses the right rule at $v_i$, hence the LR-drawing of $T_r(i)$ in $\Gamma$ has width at most $\omega_r\leq \omega-1$. Further, since $\Gamma$ uses the left rule at $v_j$, the LR-drawing of $L_j$ in $\Gamma$ has width at most $\omega_l\leq \omega-1$; since tree $T_r(i)$ is a subtree of $L_j$, for $i=j+2,\dots,q$, it also admits an LR-drawing with width at most $\omega_l$. Hence, for $i=2,\dots,q$ with $i\neq j,j+1$, tree $T_r(i)$ admits a flat star-shaped drawing $\Pi_r(i)$ with width at most $4\omega-4$. We now place all these drawings together. 

\begin{itemize}
\item For $i=2,\dots,j-2$, set $x(s_i)=1$ and place $\Pi_r(i)$ so that $B_l(\Pi_r(i))$ is on the line $x=2$ and so that the root of $T_r(i)$ is on the same horizontal line as $s_i$; for $i=2,\dots,j-3$, set $y(s_i)$ so that the bottom side of the smallest axis-parallel rectangle containing $s_i$ and $\Pi_r(i)$ is one unit above the top side of the smallest axis-parallel rectangle containing $s_{i+1}$ and $\Pi_r(i+1)$. This part of the construction is vacuous if $j\leq 3$ as in Fig.~\ref{fig:starshaped-weak-inductive}(c). 
\item Place $\Pi_r(j-1)$ so that $B_l(\Pi_r(j-1))$ is on the line $x=2$ and, if $j\geq 4$, so that the bottom side of the smallest axis-parallel rectangle containing $s_{j-2}$ and $\Pi_r(j-2)$ is one unit above $B_t(\Pi_r(j-1))$. 
\item Rotate $\Omega_D$ by $180^{\circ}$ and place it so that $B_l(\Omega_D)$ is on the line $x=2$ and $B_t(\Omega_D)$ is one unit below the smallest axis-parallel rectangle containing $s_{j-2}$, $\Pi_r(j-2)$, and $\Pi_r(j-1)$. 
\item Set $x(v_{j-1})=1$ and place $v_{j-1}$ one unit below the bottom side of the smallest axis-parallel rectangle containing $s_{j-2}$, $\Pi_r(j-2)$, $\Pi_r(j-1)$, and $\Omega_D$; further, set $x(v_{j})=1$, $y(v_{j})=y(v_{j-1})-1$, $x(v_{j+1})=4\omega$, and $y(v_{j+1})=y(v_{j})$.
\item Place $\Omega_C$ so that $B_l(\Omega_C)$ is on the line $x=2$ and $B_t(\Omega_C)$ is one unit below $v_j$.
\item Finally, for $i=j+1,\dots,q$, set $x(s_i)=1$ and place $\Pi_r(i)$ so that $B_l(\Pi_r(i))$ is on the line $x=2$ with the root of $T_r(i)$ on the same horizontal line as $s_i$; also, set $y(s_i)$ so that the bottom side of the smallest axis-parallel rectangle containing $s_{i-1}$ and $\Pi_r(i-1)$ (or containing $v_{j}$ and $\Omega_C$ if $i=j+1$) is one unit above the top side of the smallest axis-parallel rectangle containing $s_{i}$ and $\Pi_r(i)$. 
\end{itemize}

This completes the construction of $\Pi_L$. If $j=2$, then $r_T$ has been drawn in $\Pi_L$; hence, we obtain a drawing $\Pi$ of $T$ by placing $\Pi_R$ and $\Pi_L$ so that $B_b(\Pi_R)$ is one unit above $B_t(\Pi_L)$. If $j\geq 3$, then $r_T$ has not been drawn in $\Pi_L$; hence, we obtain $\Pi$ by placing $r_T$, $\Pi_R$, and $\Pi_L$ so that $x(r_T)=1$ and so that $r_T$ is one unit below $B_b(\Pi_R)$ and one unit above $B_t(\Pi_L)$. 

The only grid columns intersecting $\Pi$ are the lines $x=i$ with $i=1,\dots,4\omega$. Indeed, the nodes of the leftmost and rightmost path of $T$ lie on the line $x=1$, while $v_{j+1}$ lies on the line $x=4\omega$. Drawings $\Pi_l(i)$ and $\Pi_r(i)$ have the left sides of their bounding boxes on the line $x=2$ and have width at most $4\omega-4$; finally, drawings $\Omega_C$ and $\Omega_D$ have the left sides of their bounding boxes on the line $x=2$ and have width at most $4\omega-2$. It follows that the width of $\Pi$ is $4\omega$.

The flat property is clearly satisfied by $\Pi$. That $\Pi$ is a star-shaped drawing can be proved by exploiting the same arguments as in the proof that $\Omega$ is a star-shaped drawing. In particular, $v_{j-1}$ sees all the nodes of its left-right path since it is placed below $B_b(\Omega_D)$ and to the left of $B_l(\Omega_D)$, since $\Omega_D$ is rotated by $180^{\circ}$, and since $\Omega_D$ satisfies the bell-like property. This concludes the proof of Lemma~\ref{le:weak-star-shaped}.

Since points $p^*_u$ and $p^*_v$ can be chosen in any bell-like or flat star-shaped drawing $\Gamma$ so that the smallest axis-parallel rectangle containing $p^*_u$, $p^*_v$, and $\Gamma$ has asymptotically the same area as $\Gamma$, it follows by Lemmata~\ref{le:star-shaped-correspondence} and~\ref{le:weak-star-shaped} that, if an ordered rooted binary tree $T$ admits an LR-drawing with width $\omega$, then the outerplanar graph $T$ is the dual tree of admits an outerplanar straight-line drawing with width $O(\omega)$ and area $O(n\cdot \omega)$. 


\subsection{Star-Shaped Drawings with $O\left(2^{\sqrt{2 \log_2 n}} \sqrt{\log n}\right)$ Width} \label{se:strong-star-shaped}

In this section we show that every $n$-node ordered rooted binary tree $T$ admits a star-shaped drawing with height $O(n)$ and width $O\left(2^{\sqrt{2 \log_2 n}} \sqrt{\log n}\right)$. Similarly to the previous section, we show two different algorithms to construct star-shaped drawings of $T$. The first one, which is called {\em strong bell-like algorithm}, constructs a bell-like star-shaped drawing of $T$. The second one, which is called {\em strong flat algorithm}, constructs a flat star-shaped drawing of $T$. Throughout the section, we denote by $f(n)$ the maximum width of a drawing of an $n$-node ordered rooted binary tree constructed by means of any of these algorithms. Both algorithms are parametric, with respect to a parameter $A<n$ to be fixed later. Further, both algorithms work by induction on $n$ and exploit a structural decomposition of $T$ due to Chan et al.~\cite{c-anlabdbt-99,c-nlabdbt-02,cgkt-oaar-02}, for which we include a proof, for the sake of completeness. See Fig.~\ref{fig:decomposition-lemma}.

\begin{lemma} \label{le:decomposition} (Chan et al.~\cite{c-anlabdbt-99,c-nlabdbt-02,cgkt-oaar-02})
There exists a path $P=(v_1,\dots,v_k)$ in $T$ such that: (i) $v_1=r_T$; (ii) the subtree of $T$ rooted at $v_k$ has at least $n-A$ nodes; and (iii) each subtree of $v_k$ has less than $n-A$ nodes.
\end{lemma} 

\begin{proof}
Let $v_1=r_T$. Suppose that $P$ has been constructed up to a node $v_j$, for some $j\geq 1$, such that the subtree of $T$ rooted at $v_j$ has at least $n-A$ nodes. If a child of $v_j$ is the root of a subtree of $T$ with at least $n-A$ nodes, then let $v_{j+1}$ be that child. Otherwise, $k=j$ terminates the definition of $P$.
\end{proof}

\begin{figure}[htb]
	\centering
	\includegraphics[page=1,width=.3\textwidth]{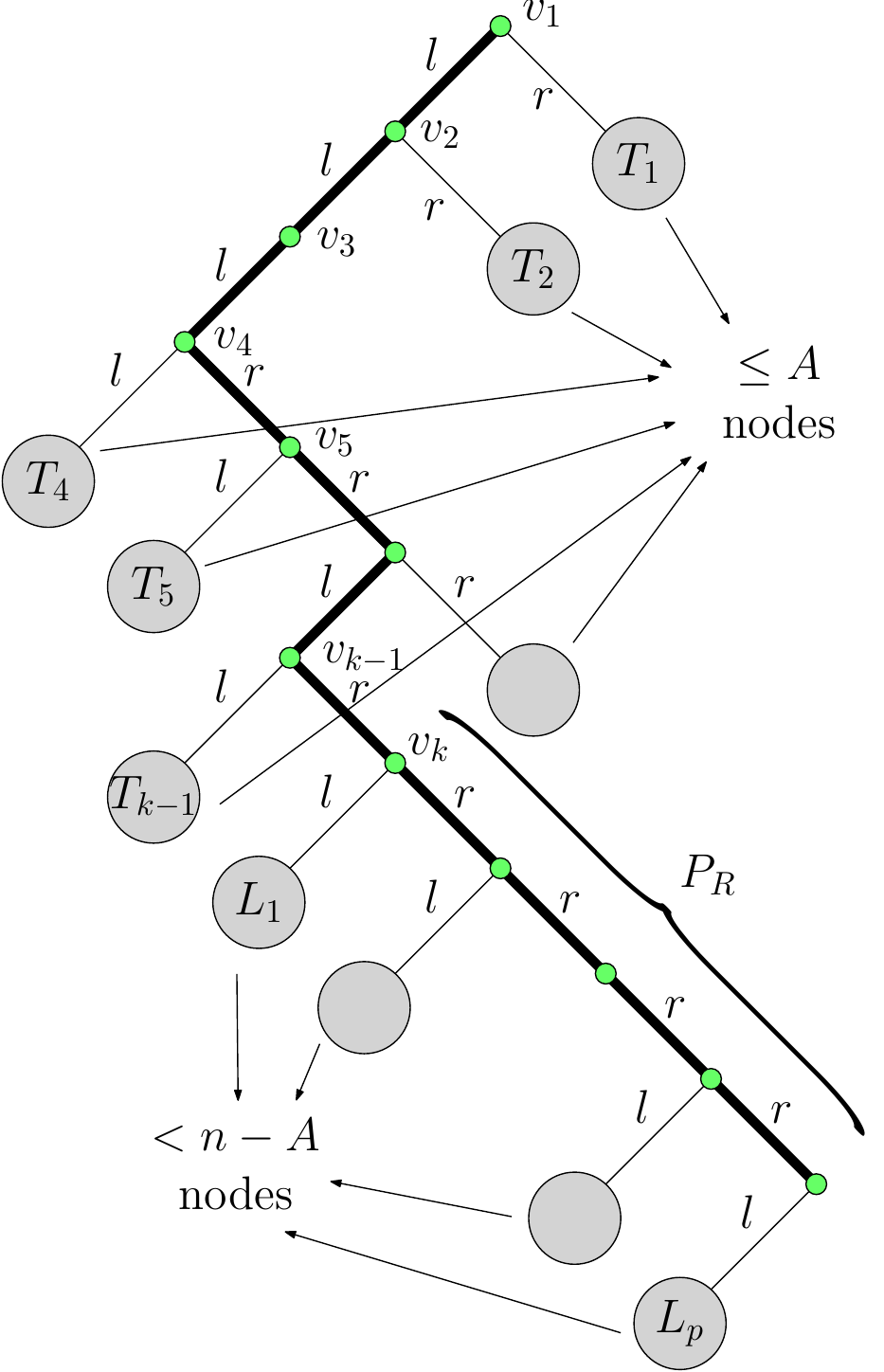}
	\caption{An illustration for the structural decomposition of $T$ exploited in Section~\ref{se:strong-star-shaped}. The tree in this example has $3$ switches, the first of which is triple $(v_3,v_4,v_5)$.}
	\label{fig:decomposition-lemma}
\end{figure}

We will say that the path $P$ is the {\em spine} of $T$. For the sake of the simplicity of the algorithm's description, we will assume that $k>1$ (the case in which $k=1$ is easy to handle) and that $v_k$ is the right child of $v_{k-1}$ (the case in which $v_k$ is the left child of $v_{k-1}$ is symmetric). For $i=1,\dots,k-1$, denote by $T_i$ the subtree of $T$ rooted at the child of $v_i$ not in $P$. Let $P_R$ be the rightmost path of the subtree of $T$ rooted at $v_k$ and let $L_1,\dots,L_p$ be the subtrees of $P_R$. Notice that each tree $T_i$ has at most $A$ nodes (by condition (ii) of Lemma~\ref{le:decomposition}) and each tree $L_i$ has less than $n-A$ nodes (by condition (iii) of Lemma~\ref{le:decomposition}). Let a {\em switch} of the spine $P$ be a triple $(v_i,v_{i+1},v_{i+2})$ with $i\leq k-2$ such that: (i) $v_{i+1}$ is the left child of $v_i$ and $v_{i+2}$ is the right child of $v_{i+1}$; or (ii) $v_{i+1}$ is the right child of $v_i$ and $v_{i+2}$ is the left child of $v_{i+1}$. Let $s$ be number of switches of $P$. For $i=1,\dots,s$, let $\pi(i)$ be such that $(v_{\pi(i)},v_{\pi(i)+1},v_{\pi(i)+2})$ is the $i$-th switch of $P$. Note that $\pi(i+1)\geq \pi(i)+1$, for $i=1,\dots,s-1$. 

The strong flat algorithm uses different constructions for the case in which $s\leq 7$ and the case in which $s\geq 8$. Further, the strong bell-like algorithm uses different constructions for the case in which $s\leq 4$ and the case in which $s\geq 5$. We start by describing the construction which is used by the strong flat algorithm if $s\leq 7$. 

{\bf Strong flat algorithm with $\bf s\leq 7$.} This is the easiest case of the recursive algorithm. The spine $P$, together with the leftmost and rightmost paths of $T$ and of certain subtrees of $T$, is going to be drawn on a set of at most $s+1$ grid columns. In fact, the first vertices of the spine (up to $v_{\pi(1)+1}$) are drawn on the line $x=0$; then the drawing moves at most one grid column to the right at every switch of the spine. The resulting drawing of the spine $P$ has a ``zig-zag'' shape, where each part of this zig-zag is a subpath of $P$ drawn on a single grid column from top to bottom or vice versa. We now formally describe this construction; the description uses induction on $s$.

In the base case we have ${\bf s=0}$; refer to Fig.~\ref{fig:strong-flat-s5}(a). Since $s=0$, it follows that $P$ has no switches, hence $v_{i+1}$ is the right child of $v_i$, for $i=1,\dots,k-1$, given that $v_k$ is the right child of $v_{k-1}$ by hypothesis. Let $P_0$ be the leftmost path of the left subtree of $T$. Recursively construct a flat star-shaped drawing of the trees $T_2,\dots,T_{k-1}$, of the trees $L_1,\dots,L_p$, and of the subtrees of $P_0$. 

For $i=2,\dots,k-1$, augment the recursively constructed drawing of $T_i$ by placing the parent of $r_{T_i}$ one unit to the left of $r_{T_i}$; similarly augment the recursively constructed drawings of the trees $L_1,\dots,L_p$, and of the subtrees of $P_0$. Further, construct a drawing (consisting of a single point) of every node that has not been drawn yet (these are the nodes of the leftmost path of $T$ with no right child and the nodes of the rightmost path of $T$ with no left child). We now place all these drawings together. 

First, set the $x$-coordinate of every node in the leftmost and rightmost path of $T$ to be $0$. Since each tree that has been individually drawn contains a node in the leftmost or rightmost path of $T$ (due to the above described augmentation of each recursively constructed drawing), this assignment determines the $x$-coordinate of every node of $T$.  

Second, we assign a $y$-coordinate to every node of $T$. This is done so that every grid row contains a node or intersects a subtree. Rather than providing explicit $y$-coordinates, we establish a total order $\sigma$ for a set that contains one node for each individually drawn tree; then a $y$-coordinate assignment is obtained by forcing, for any two nodes $u_j$ and $u_{j+1}$ that are consecutive in $\sigma$, the top side of the bounding box of the drawing comprising $u_j$ to be one unit below the bottom side of the bounding box of the drawing comprising $u_{j+1}$. Order $\sigma$ consists of the nodes of the leftmost path of $T$ in {\em reverse order} (that is, from the unique leaf to $r_T$) followed by the nodes of the rightmost path of the right subtree of $T$ in {\em straight order} (that is, from the root to the unique  leaf). This completes the construction of a drawing $\Gamma$ of $T$.

  \begin{figure}[!tb]
    \centering
    \hfill
	\subfloat[]{
	\includegraphics[scale=0.9]{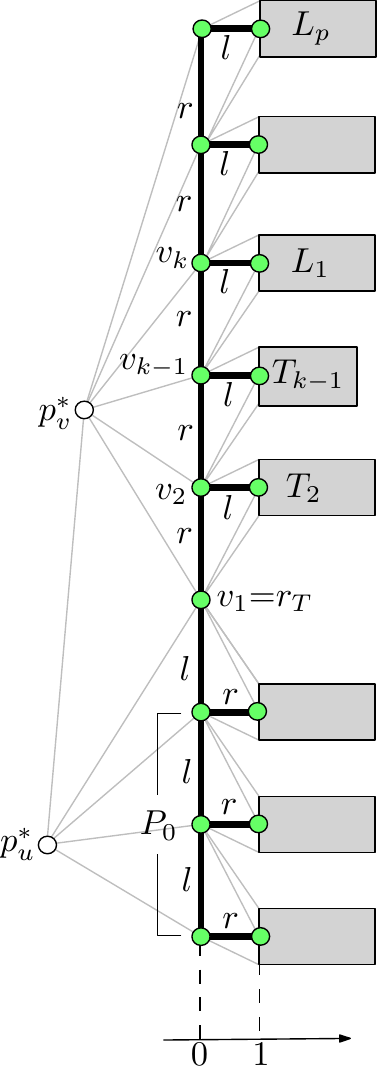}}
    \hfill
	\subfloat[]{
	\includegraphics[scale=0.9]{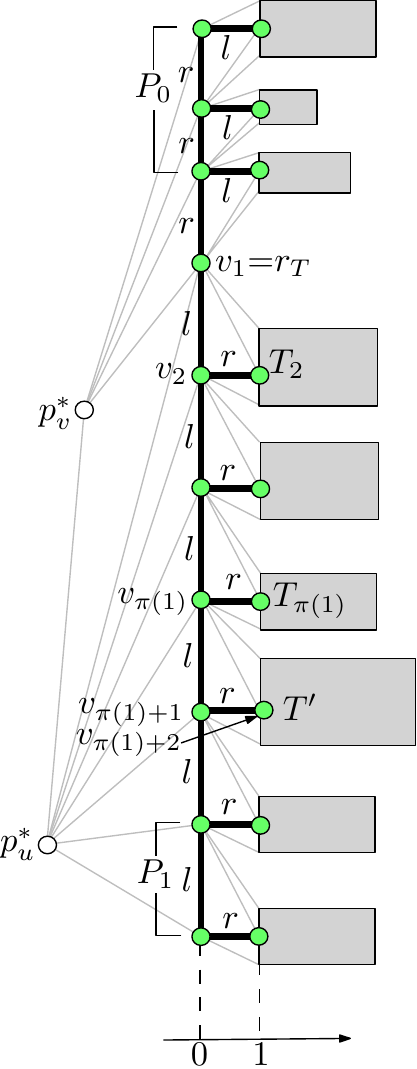}}
    \hfill
	\subfloat[]{
	\includegraphics[scale=0.9]{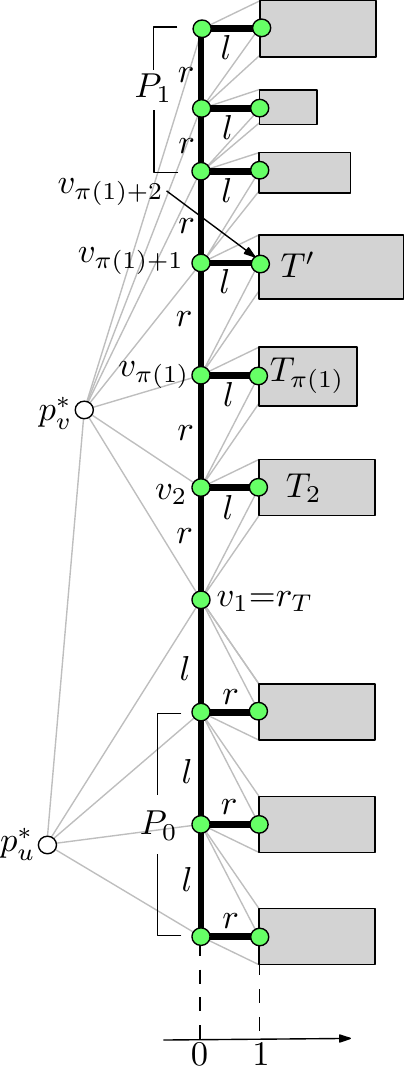}}
    \hfill \
	\caption{Illustration for the strong flat algorithm when $s\leq 7$. (a) The case $s=0$. (b) The case $1\leq s\leq 7$ with $s$ odd. (c) The case $1\leq s\leq 7$ with $s$ even.}
    \label{fig:strong-flat-s5}
  \end{figure}
  
In the inductive case we have ${\bf 1\leq s\leq 7}$. By hypothesis, we have that $v_k$ is the right child of $v_{k-1}$; hence, if $s$ is odd then $v_{i+1}$ is the left child of $v_i$, for $i=1,\dots,\pi(1)$, otherwise $v_{i+1}$ is the right child of $v_i$, for $i=1,\dots,\pi(1)$. We formally describe the construction for the case in which $s$ is odd, which is illustrated in Fig.~\ref{fig:strong-flat-s5}(b). The construction for the other case is symmetric (see Fig.~\ref{fig:strong-flat-s5}(c)). 

Let $P_0$ be the rightmost path of the right subtree of $T$, let $P_1$ be the leftmost path of the left subtree of $v_{\pi(1)+1}$, and let $T'$ be the subtree of $T$ rooted at $v_{\pi(1)+2}$. Recursively construct a flat star-shaped drawing of trees $T_2,\dots,T_{\pi(1)}$, of the subtrees of $P_0$, and of the subtrees of $P_1$. Further, notice that the subpath of $P$ contained in $T'$ has either $s-1$ or $s-2$ switches (indeed, it has $s-2$ switches if $v_{\pi(2)}=v_{\pi(1)+1}$ and it has $s-1$ switches otherwise). Then the drawing of $T'$ can be constructed inductively. We stress the fact that the spine is not recomputed for $T'$ according to Lemma~\ref{le:decomposition}, but rather the construction of the drawing of $T'$ is completed by using the subpath of $P$ between $v_{\pi(1)+2}$ and $v_k$ as the spine for $T'$. 

For $i=2,\dots,\pi(1)$, augment the recursively constructed drawing of $T_i$ by placing the parent of $r_{T_i}$ one unit to the left of $r_{T_i}$; similarly augment the drawings of $T'$ and of the subtrees of $P_0$ and $P_1$. Further, construct a drawing (consisting of a single point) of every node that has not been drawn yet. We now place all these drawings together.

First, set the $x$-coordinate of every node in the leftmost and rightmost paths of $T$ to be $0$. This determines the $x$-coordinate of every node of $T$. Second, we establish a total order $\sigma$ for a set that contains one node for each individually drawn tree; then a $y$-coordinate assignment is obtained by forcing, for any two nodes $u_j$ and $u_{j+1}$ that are consecutive in $\sigma$, the top side of the bounding box of the drawing comprising $u_j$ to be one unit below the bottom side of the bounding box of the drawing comprising $u_{j+1}$. Order $\sigma$ consists of the nodes of $P_1$ in reverse order, then of the nodes $v_{\pi(1)+1},v_{\pi(1)},v_{\pi(1)-1},\dots,v_1$, and then of the nodes of $P_0$ in straight order. This completes the construction of a drawing $\Gamma$ of $T$. We get the following.

\begin{lemma} \label{le:strong-flat-05}
Suppose that $s\leq 7$. Then the strong flat algorithm constructs a flat star-shaped drawing whose height is at most $n$ and whose width is at most $8+\max\{f(A),f(n-A)\}$.
\end{lemma} 

\begin{proof}
It is readily seen that $\Gamma$ is star-shaped and flat. In particular, consider any node $u$ in the leftmost or rightmost path of $T$. By construction, $u$ is on the line $x=0$. Further, all the nodes that are not adjacent to $u$ and that are in the left-right path or in the right-left path of $u$ lie on the line $x=1$ (indeed, all such nodes are in the leftmost or rightmost paths of some subtrees of $T$ for which flat star-shaped drawings have been recursively constructed and embedded with the left sides of their bounding boxes on the line $x=1$); hence $u$ sees all such nodes. That any node that is not in the leftmost or rightmost path of $T$ sees all the non-adjacent nodes in its left-right path and in its right-left path comes from induction. Drawing $\Gamma$ has height at most $n$ since any horizontal grid line intersecting $\Gamma$ passes through a node in the leftmost or rightmost path of $T$ or intersects a recursively constructed drawing. Further, it can be proved by induction on $s$ that the width of $\Gamma$ is at most $s+1+\max\{f(A),f(n-A)\}$. Indeed, if $s=0$ then all the subtrees that are drawn by a recursive application of the strong flat algorithm have either at most $A$ nodes or at most $n-A$ nodes and have the left side of their bounding boxes on the line $x=1$; this suffices to prove the statement, since no node has an $x$-coordinate that is smaller than $0$. If $s>0$, then the statement follows inductively, given that the spine of $T'$ has at most $s-1$ switches and no node of $T'$ has an $x$-coordinate that is smaller than $1$.  
\end{proof}

We now describe the strong bell-like algorithm for the case in which $s\leq 4$.

{\bf Strong bell-like algorithm with $\bf s\leq 4$.} In this case the leftmost or the rightmost path of $T$, depending on whether $s$ is odd or even, respectively, is going to be drawn on a single grid column; in particular, this grid column is the leftmost or the rightmost grid column intersecting the drawing, depending on whether $s$ is odd or even, respectively. Similarly to the strong flat algorithm, the spine $P$, together with the leftmost and rightmost paths of $T$ and of certain subtrees of $T$, is going to be drawn on a set of $s+1$ grid columns; also, $P$ is going to have a ``zig-zag'' shape. We now formally describe this construction; the description uses induction on $s$.

In the base case we have ${\bf s=0}$; refer to Fig.~\ref{fig:strong-bell-s5}(a). Then $v_{i+1}$ is the right child of $v_i$, for $i=1,\dots,k-1$, given that $v_k$ is the right child of $v_{k-1}$ by hypothesis. Recursively construct a bell-like drawing $\Gamma_1$ of $T_1$; also, by means of the strong flat algorithm, construct a flat star-shaped drawing of the trees $T_2,\dots,T_{k-1}$ and of the trees $L_1,\dots,L_p$. Rotate each of the constructed flat star-shaped drawings by $180^{\circ}$. 

For $i=2,\dots,k-1$, augment the drawing of $T_i$ by placing the parent of $r_{T_i}$ one unit to the right of $r_{T_i}$; similarly augment the drawings of the trees $L_1,\dots,L_p$. Augment $\Gamma_1$ by placing $v_1$ one unit above $B_t(\Gamma_1)$ and one unit to the right of $B_r(\Gamma_1)$. Further, construct a drawing (consisting of a single point) of every node that has not been drawn yet (these are the nodes of the rightmost path of $T$ with no left child). We now place all these drawings together. 

First, set the $x$-coordinate of every node in the rightmost path of $T$ to be $0$. This determines the $x$-coordinate of every node of $T$. Second, we establish a total order $\sigma$ for a set that contains one node for each individually drawn tree; then a $y$-coordinate assignment is obtained by forcing, for any two nodes $u_j$ and $u_{j+1}$ that are consecutive in $\sigma$, the top side of the bounding box of the drawing comprising $u_j$ to be one unit below the bottom side of the bounding box of the drawing comprising $u_{j+1}$. Order $\sigma$ consists of the nodes of the rightmost path of $T$ in reverse order. This completes the construction of a drawing $\Gamma$ of $T$. 

\begin{figure}[!tb]
    \centering
    \hfill
	\subfloat[]{
	\includegraphics[scale=0.9]{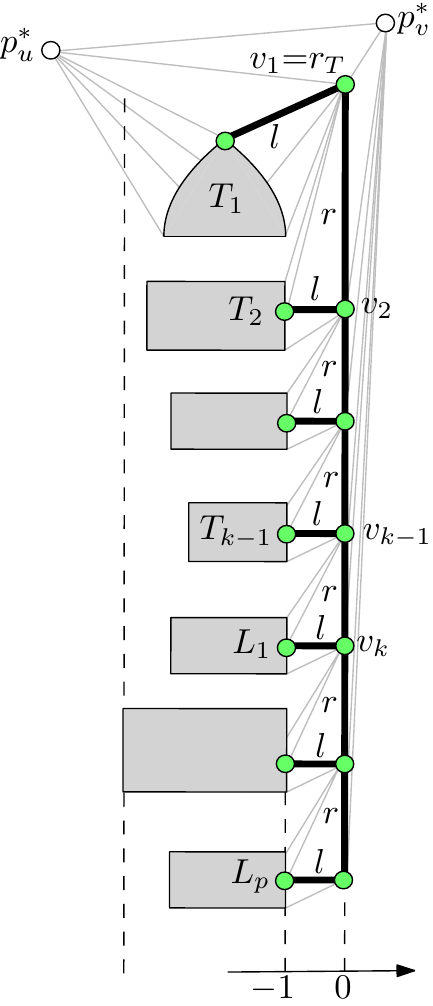}}
    \hfill
	\subfloat[]{
	\includegraphics[scale=0.9]{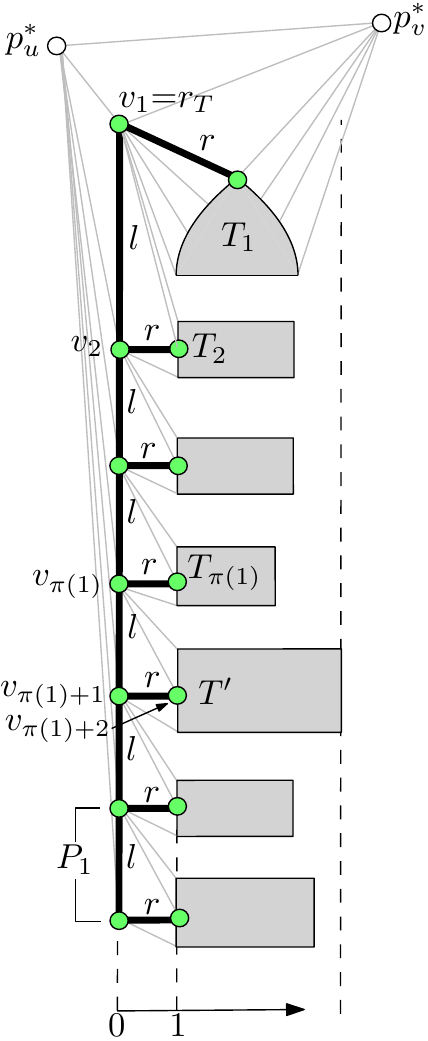}}
    \hfill
	\subfloat[]{
	\includegraphics[scale=0.9]{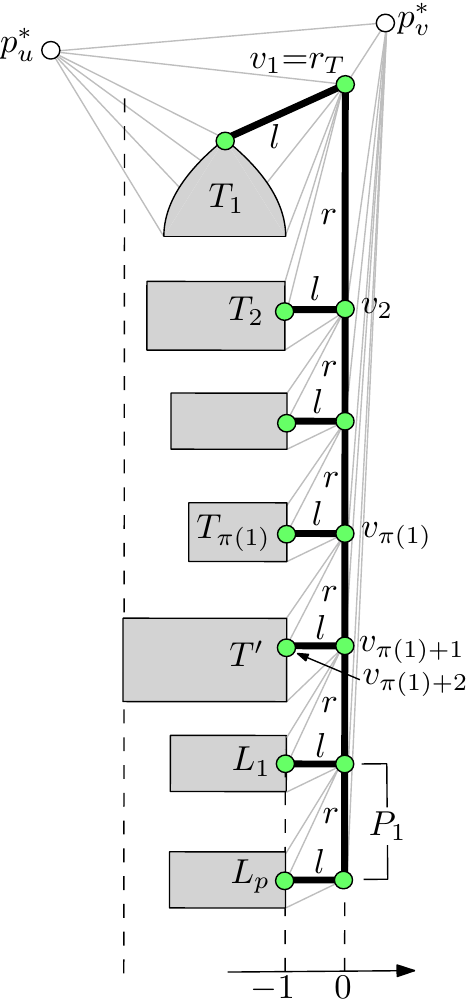}}
    \hfill \
	\caption{Illustration for the strong bell-like algorithm when $s\leq 4$. (a) The case $s=0$. (b) The case $1\leq s\leq 4$ with $s$ odd. (c) The case $1\leq s\leq 4$ with $s$ even.}
    \label{fig:strong-bell-s5}
  \end{figure}
  
In the inductive case we have ${\bf 1\leq s\leq 4}$. By hypothesis, we have that $v_k$ is the right child of $v_{k-1}$; hence, if $s$ is odd (even) then $v_{i+1}$ is the left (resp.\ right) child of $v_i$, for $i=1,\dots,\pi(1)$. We first describe the construction for the case in which $s$ is odd, which is illustrated in Fig.~\ref{fig:strong-bell-s5}(b). 

Let $P_1$ be the leftmost path of the left subtree of $v_{\pi(1)+1}$ and let $T'$ be the subtree of $T$ rooted at $v_{\pi(1)+2}$. Recursively construct a bell-like drawing $\Gamma_1$ of $T_1$; also, by means of the strong flat algorithm, construct a flat star-shaped drawing of the trees $T_2,\dots,T_{\pi(1)}$ and of the subtrees of $P_1$. Further, notice that the part of $P$ contained in $T'$ has either $s-1$ or $s-2$ switches (indeed, it has $s-2$ switches if $v_{\pi(2)}=v_{\pi(1)+1}$ and it has $s-1$ switches otherwise). Then a flat star-shaped drawing of $T'$ is constructed by means of the strong flat algorithm; we stress the fact that the spine is not recomputed for $T'$ according to Lemma~\ref{le:decomposition}, but rather the construction of the drawing of $T'$ is completed by using the subpath of $P$ between $v_{\pi(1)+2}$ and $v_k$ as the spine for $T'$. 

For $i=2,\dots,\pi(1)$, augment the drawing of $T_i$ by placing the parent of $r_{T_i}$ one unit to the left of $r_{T_i}$; similarly augment the drawings of $T'$ and of the subtrees of $P_1$. Augment $\Gamma_1$ by placing $v_1$ one unit above $B_t(\Gamma_1)$ and one unit to the left of $B_l(\Gamma_1)$. Further, construct a drawing (consisting of a single point) of every node that has not been drawn yet (these are the nodes of the leftmost path of $T$ with no right child). These drawings are placed together as in the case in which $s=0$. In particular, set the $x$-coordinate of every node in the leftmost path of $T$ to be $0$, thus determining the $x$-coordinate of every node of $T$. Further, the $y$-coordinate assignment is such that the top side of the bounding box of the drawing comprising a node of the leftmost path of $T$ is one unit below the bottom side of the bounding box of the drawing comprising the parent of that node. This completes the construction of a drawing $\Gamma$ of $T$. 

The case in which $s$ is even, which is illustrated in Fig.~\ref{fig:strong-bell-s5}(c), is symmetric to the previous one and very similar to the case $s=0$. In particular, the rightmost path of $T$ is drawn on the rightmost grid column intersecting the drawing. Further, each recursively constructed flat star-shaped drawing of a subtree of the rightmost path of $T$ has to be rotated by $180^{\circ}$ and placed so that its root is one unit to the left of its parent. We get the following.

\begin{lemma} \label{le:strong-bell-like-05}
Suppose that $s\leq 4$. Then the strong bell-like algorithm constructs a bell-like star-shaped drawing whose height is at most $n$ and whose width is at most $5+\max\{f(A),f(n-A)\}$. 
\end{lemma} 

\begin{proof}
Assume that $s$ is odd; the case in which $s$ is even is symmetric.
 
It is readily seen that $\Gamma$ is star-shaped. In particular, it can be proved similarly to the proof of Lemma~\ref{le:strong-flat-05} that every node different from $r_T$ sees all the nodes in its left-right path and in its right-left path that are not adjacent to it, and that $r_T$ sees all the nodes in its left-right path that are not adjacent to it. Further, $r_T$ sees all the nodes in its right-left path that are not adjacent to it, since all such nodes are in the leftmost path of $T_1$, since the drawing $\Gamma_1$ of $T_1$ is bell-like, and since $r_T$ is one unit above $B_t(\Gamma_1)$ and one unit to the left of $B_l(\Gamma_1)$. 

Drawing $\Gamma$ is also bell-like. Indeed: (i) $r_T$ lies on $B_t(\Gamma)$ by construction; (ii) the nodes of the leftmost path of $T$ lie on the line $x=0$ in decreasing order of $y$-coordinates from $r_T$ to the unique leaf, and no other node of $T$ has an $x$-coordinate smaller than $1$; (iii) the drawing $\Gamma_1$ of $T_1$ is bell-like, $r_T$ is one unit above and at least one unit to the left of $r_{T_1}$, and every node of $T$ different from $r_T$ and not in $T_1$ is below $B_b(\Gamma_1)$. These statements imply that any point $p^*_u$ above $B_t(\Gamma)$ and to the left of $B_l(\Gamma)$ and any point $p^*_v$ above $B_t(\Gamma)$ and to the right of $B_r(\Gamma)$ satisfy Property 4 of a star-shaped drawing. 

Drawing $\Gamma$ has height at most $n$ since any horizontal grid line intersecting $\Gamma$ passes through a node on the leftmost path of $T$ or intersects a recursively constructed drawing. Concerning the width of $\Gamma$, note that the only subtree $T_1$ that is drawn by a recursive application of the strong bell-like algorithm has at most $A$ nodes (or at most $n-A$ nodes if $k$ were equal to $1$) and has the left side of its bounding box on the line $x=1$, while no node of $T$ has an $x$-coordinate that is smaller than $0$. The argument for the subtrees that are recursively drawn by means of the strong flat algorithm is analogous to the one in the proof of Lemma~\ref{le:strong-flat-05}.
\end{proof}

In general, it might hold that $s=\Omega(A)$; hence, if the strong flat algorithm and the strong bell-like algorithm used the constructions described above for every value of $s$, then recurring over the trees $L_1,\dots,L_p$ one would get a drawing with $\Omega(n)$ width. For this reason, the strong flat algorithm and the strong bell-like algorithm exploit different geometric constructions when $s\geq 8$ and when $s\geq 5$, respectively. We now describe the strong bell-like algorithm in the case in which $s\geq 5$. 

{\bf Strong bell-like algorithm with $\bf s\geq 5$.} The general idea of the upcoming construction is the following. We would like to construct a bell-like star-shaped drawing $\Gamma$ whose width is given by either (i) a constant plus the width of a recursively constructed drawing of a tree with at most $n-A$ nodes, or (ii) a constant plus the widths of the recursively constructed drawings of two trees, each with at most $A$ nodes. Part of the construction we are going to show is very similar to the construction of the (non-strong) bell-like algorithm from Section~\ref{se:weak-star-shaped}: Starting from $r_T$, we draw the spine $P$ of $T$ on two adjacent grid columns, with the left subtrees of $P$ to the left of $P$ and with the right subtrees of $P$ to the right of $P$ (note that the width of this part of $\Gamma$ is a constant plus the widths of the recursively constructed drawings of two trees, each with at most $A$ nodes). Before reaching $v_k$, however, the construction changes significantly. In particular, the drawing of $P$ touches $B_r(\Gamma)$ and then continues on the grid column one unit to the left of $B_r(\Gamma)$. The remainder of $P$, including $v_k$ and together with the rightmost path $P_R$ of the subtree of $T$ rooted at $v_k$, is drawn entirely on that grid column, with its subtrees to the left of it (note that the width of this part of $\Gamma$ is a constant plus the width of a recursively constructed drawing of a tree with at most $n-A$ nodes). 

In order to guarantee that $\Gamma$ is a bell-like star-shaped drawing, it is vital that the drawings of $T_1$ and $T_{\pi(1)+1}$ are bell-like. This requirement can be easily met if the parents $v_1$ and $v_{\pi(1)+1}$ of the roots of these subtrees occur in the first part of $P$, which is drawn on two adjacent grid columns. On the other hand, if $v_1$ and $v_{\pi(1)+1}$ occurred in the second part of $P$, then the requirement on $T_1$ and $T_{\pi(1)+1}$ would conflict with the geometric constraints our construction needs to satisfy in order to place the final part of $P$, together with $P_R$, on the grid column one unit to the left of $B_r(\Gamma)$. This is the reason why we need the spine to have some number of switches (in fact at least $5$ switches).   

  \begin{figure}[!tb]
  	\centering
  	\hfill
  	\subfloat[]{
	\includegraphics[scale=0.9]{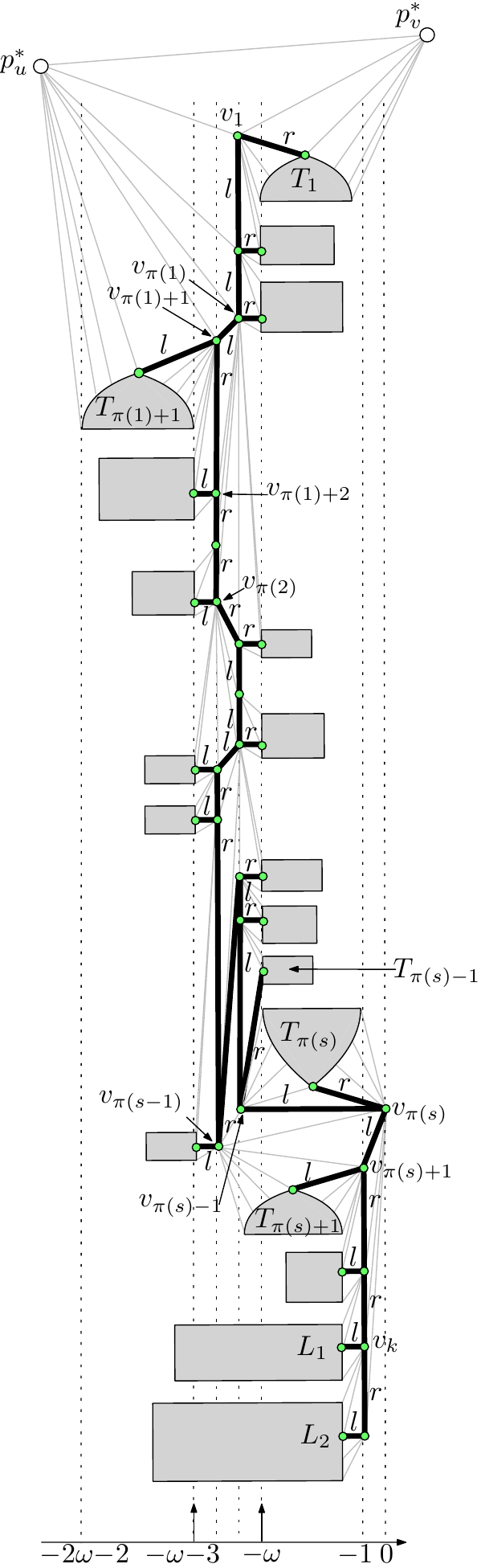}}
  	\hfill
  	\subfloat[]{
	\includegraphics[scale=0.9]{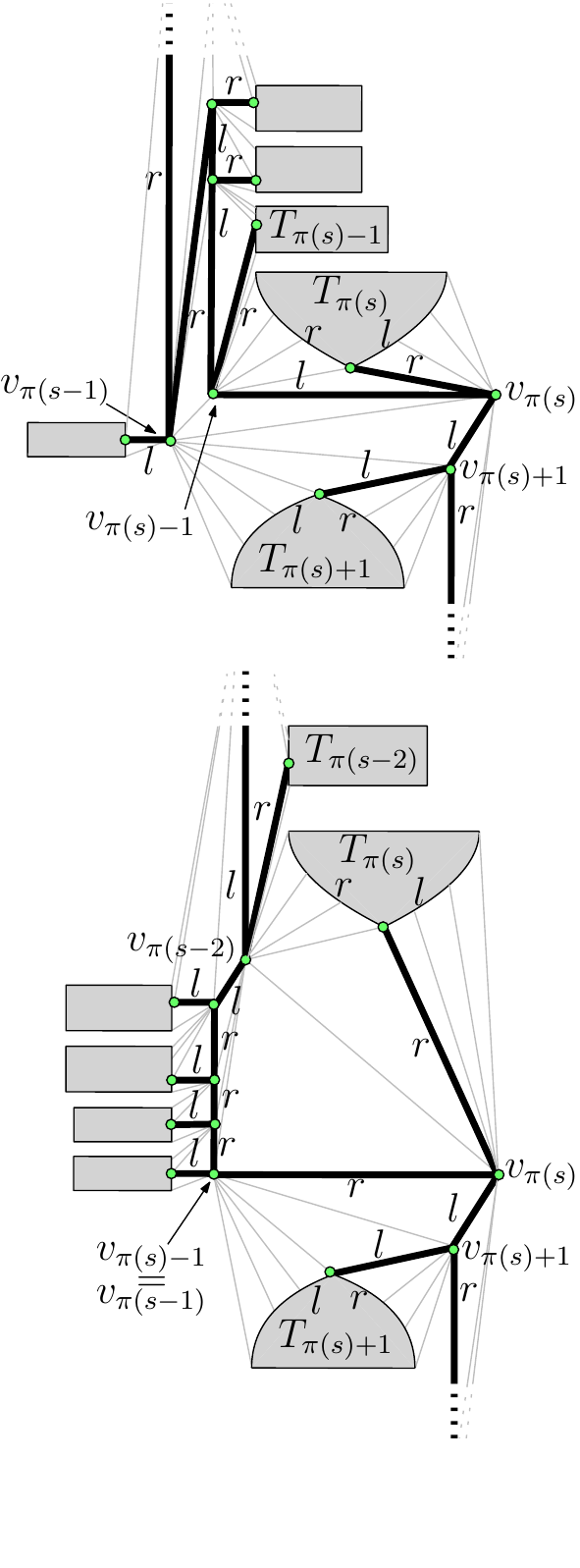}}
  	\hfill
  	\subfloat[]{
	\includegraphics[scale=0.9]{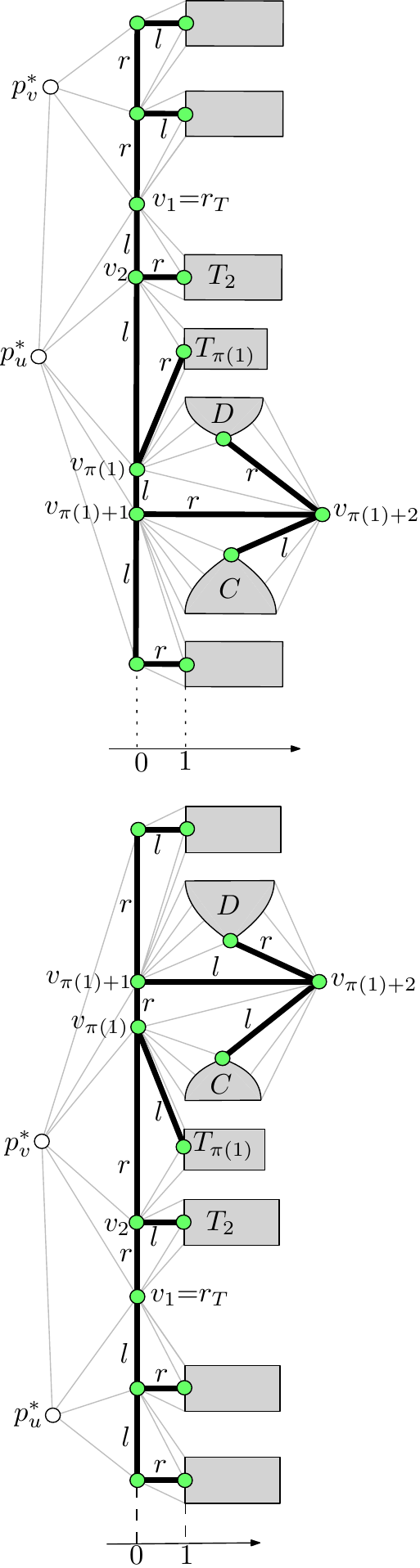}}
  	\hfill 
  	\caption{(a) Illustration for the strong bell-like algorithm when $s\geq 5$. (b) A closer look at the cases in which $\pi(s-1)<\pi(s)-1$ (top) or $\pi(s-1)=\pi(s)-1$ (bottom). (c) Illustration for the strong flat algorithm when $s\geq 8$, in the case in which $v_2$ is the left child of $v_1$ (top) or the right child of $v_1$ (bottom).}
  	\label{fig:strong-s6}
  \end{figure}

We now detail our construction. Refer to Fig.~\ref{fig:strong-s6}(a). First, we draw some subtrees recursively. We use the strong bell-like algorithm to construct a bell-like star-shaped drawing of $T_1$, of $T_{\pi(1)+1}$, of $T_{\pi(s)}$, and of $T_{\pi(s)+1}$. Further, we use the strong flat algorithm to construct a flat star-shaped drawing of every subtree $T_j$ of $P$ such that $2\leq j\leq k-1$ with $j\notin \{\pi(1)+1,\pi(s),\pi(s)+1\}$. Let $\omega$ denote the maximum width among the constructed drawings of the trees $T_j$, with $1\leq j\leq k-1$; notice that any such a subtree has at most $A$ nodes. Finally, we use the strong flat algorithm to construct flat star-shaped drawings of the trees $L_1,\dots,L_p$, which have at most $n-A$ nodes. 

We now describe an $x$-coordinate assignment for the nodes of $T$; for the part of $P$ up to $v_{\pi(s)-1}$ (that is, up to one node before the last switch of $P$), this assignment is done similarly to the (non-strong) bell-like algorithm from Section~\ref{se:weak-star-shaped} (for technical reasons, however, the nodes of $T$ are here assigned non-positive $x$-coordinates). For $i=1,\dots,\pi(s)-1$, node $v_i$ is placed on the line $x=-\omega-2$ or $x=-\omega-1$, depending on whether $v_{i+1}$ is the right or the left child of $v_i$, respectively. Further, for $i=1,\dots,\pi(s)-1$, the recursively constructed drawing of $T_i$ is assigned $x$-coordinates such that the left side of its bounding box is on the line $x=-\omega$ if $T_i$ is the right subtree of $v_i$, or it is first rotated by $180^{\circ}$ and then assigned $x$-coordinates so that the right side of its bounding box is on the line $x=-\omega-3$ if $T_i$ is the left subtree of $v_i$. Note that the part of $T$ to which $x$-coordinates have been assigned so far lies in the closed vertical strip $-2\omega-2\leq x\leq -1$, given that the width of the drawing of $T_i$ is at most $\omega$, for $i=1,\dots,\pi(s)-1$. Set $x(v_{\pi(s)})=0$; also set the $x$-coordinate of every node $v_i$, with $i=\pi(s)+1,\dots,k-1$, and of every node in $P_R$ to be $-1$. Rotate the drawing of $T_{\pi(s)}$ by $180^{\circ}$ and assign $x$-coordinates to it so that the left side of its bounding box is on the line $x=-\omega$. Finally, assign $x$-coordinates to the drawings of $T_{\pi(s)+1},\dots,T_{k-1},L_1,\dots,L_p$ so that the right sides of their bounding boxes are on the line $x=-2$.

We now describe a $y$-coordinate assignment for the nodes of $T$. Part of this assignment varies depending on whether $\pi(s-1)<\pi(s)-1$ (see Figs.~\ref{fig:strong-s6}(a) and~\ref{fig:strong-s6}(b) top) or $\pi(s-1)=\pi(s)-1$ (see Fig.~\ref{fig:strong-s6}(b) bottom). First, we define the $y$-coordinates of certain nodes with respect to the ones of their subtrees. We let node $v_1$ (node $v_{\pi(1)+1}$, node $v_{\pi(s)+1}$) have $y$-coordinate equal to $1$ plus the $y$-coordinate of the root of $T_1$ (resp.\ of $T_{\pi(1)+1}$, resp.\ of $T_{\pi(s)+1}$). Further, for $j=1,\dots,p$, we let the root of $L_j$ have the same $y$-coordinate as its parent. Also:

\begin{itemize}
	\item If $\pi(s-1)<\pi(s)-1$, then we let the root of $T_j$ have the same $y$-coordinate as its parent for $j=2,\dots,k-1$ with $j\notin\{\pi(1)+1,\pi(s)-1,\pi(s),\pi(s)+1\}$. 
	\item If $\pi(s-1)=\pi(s)-1$, then we let the root of $T_j$ have the same $y$-coordinate as its parent for $j=2,\dots,k-1$ with $j\notin\{\pi(1)+1,\pi(s-2),\pi(s),\pi(s)+1\}$.
\end{itemize}

We construct a drawing (consisting of a single point) of every node that has not yet been drawn, including $v_{\pi(s)}$, including $v_{\pi(s)-1}$ (if $\pi(s-1)<\pi(s)-1$), and including $v_{\pi(s-2)}$ (if $\pi(s-1)=\pi(s)-1$). Note that the $y$-coordinates of $T_{\pi(s)}$ have not been defined relatively to the one of $v_{\pi(s)}$; analogously, if $\pi(s-1)<\pi(s)-1$ (if $\pi(s-1)=\pi(s)-1$), then the $y$-coordinates of $T_{\pi(s)-1}$ (resp.\ of $T_{\pi(s-2)}$) have not been defined relatively to the one of $v_{\pi(s)-1}$ (resp.\ of $v_{\pi(s-2)}$).

We now place all these drawings together. Namely, we define a total order $\sigma$ of the nodes and subtrees of $T$ that have been individually drawn; then we can recover a $y$-coordinate assignment from $\sigma$ by interpreting it as a {\em top-to-bottom} order of the subtrees (note that, in the previously described constructions, the order $\sigma$ represented a {\em bottom-to-top} order of the subtrees), so that the bottom side of the bounding box of a subtree is one unit above the top side of the bounding box of the next subtree in $\sigma$. The order $\sigma$ starts with the nodes $v_1,v_2,\dots,v_{\pi(s-2)-1}$. 

\begin{itemize}
	\item If $\pi(s-1)<\pi(s)-1$, then the order $\sigma$ continues with $v_{\pi(s-2)},\dots,v_{\pi(s-1)-1}$, with $v_{\pi(s-1)+1},\dots,v_{\pi(s)-2}$, with $T_{\pi(s)-1}$, with $T_{\pi(s)}$, with $v_{\pi(s)-1}$ and $v_{\pi(s)}$ (which have the same $y$-coordinate), and with $v_{\pi(s-1)}$. 
	\item If $\pi(s-1)=\pi(s)-1$, then the order $\sigma$ continues with $T_{\pi(s-2)}$, with $T_{\pi(s)}$, with $v_{\pi(s-2)},\dots,v_{\pi(s-1)-1}$, and with $v_{\pi(s-1)}$ and $v_{\pi(s)}$ (which have the same $y$-coordinate). 
\end{itemize}

The order $\sigma$ terminates with the nodes $v_{\pi(s)+1},\dots,v_{k-1}$ and with the nodes of $P_R$ in straight order. This concludes the construction of the drawing $\Gamma$. We have the following.

\begin{lemma} \label{le:strong-bell-like-6}
	Suppose that $s\geq 5$. Then the strong bell-like algorithm constructs a bell-like star-shaped drawing whose height is at most $n$ and whose width is at most $3 + \max\{2f(A),f(n-A)\}$. 
\end{lemma} 

\begin{proof}
It is readily seen that $\Gamma$ is star-shaped and bell-like. Most interestingly:

\begin{itemize}
\item If $\pi(s-1)<\pi(s)-1$, then $v_{\pi(s-1)}$ sees all the nodes of its right-left path that are not adjacent to it. Indeed, the subpath $(v_{\pi(s-1)+1},\dots,v_{\pi(s)-1})$ of the right-left path of $v_{\pi(s-1)}$ is represented by a straight-line segment on the vertical line $x=-\omega-1$, which is one unit to the right of $v_{\pi(s-1)}$, so that $v_{\pi(s)-1}$ is the point of this segment with the smallest $y$-coordinate and is above $v_{\pi(s-1)}$; hence, this segment does not block the visibility between $v_{\pi(s-1)}$ and $v_{\pi(s)}$, which has the same $y$-coordinate as $v_{\pi(s)-1}$ and is to the right of it, and between $v_{\pi(s-1)}$ and $v_{\pi(s)+1}$, which is below $v_{\pi(s-1)}$. Finally, $v_{\pi(s-1)}$ sees all the nodes of the leftmost path of $T_{\pi(s)+1}$, given that the drawing of $T_{\pi(s)+1}$ is bell-like and that $v_{\pi(s-1)}$ lies to the left and above the left side and the top side of the bounding box of the drawing of $T_{\pi(s)+1}$, respectively (note that $x(v_{\pi(s-1)})=-\omega-2$, while $T_{\pi(s)+1}$ has $x$-coordinates in the range $-\omega-1\leq x\leq -2$).

\item If $\pi(s-1)=\pi(s)-1$, then $v_{\pi(s-2)}$ sees all the nodes of its left-right path that are not adjacent to it. Indeed, the subpath $(v_{\pi(s-2)+1},\dots,v_{\pi(s-1)})$ of the left-right path of $v_{\pi(s-2)}$ is represented by a straight-line segment on the vertical line $x=-\omega-2$, which is one unit to the left of $v_{\pi(s-2)}$; hence, this segment does not block the visibility between $v_{\pi(s-2)}$ and $v_{\pi(s)}$, which is to the right of $v_{\pi(s-2)}$ and below it. Finally, $v_{\pi(s-2)}$ sees all the nodes of the rightmost path of $T_{\pi(s)}$, given that the drawing of $T_{\pi(s)}$ is bell-like and is rotated by $180^{\circ}$, and that $v_{\pi(s-2)}$ lies to the left and below the left side and the bottom side of the bounding box of the drawing of $T_{\pi(s)}$, respectively. 
\end{itemize}

We remark that, if $\pi(s-1)=\pi(s)-1$, then the algorithm constructs a flat star-shaped drawing of $T_{\pi(s-2)}$ and places this drawing so that the bottom side of its bounding box is above $v_{\pi(s-2)}$, in order to ``make space'' for the drawing of $T_{\pi(s)}$. On the other hand, in order to ensure the bell-like property for $\Gamma$, the construction employs a bell-like drawing of $T_{\pi(1)+1}$. Hence, we need $\pi(1)+1$ to be smaller than $\pi(s-2)$. However, we have $\pi(1)+1\leq \pi(2)$ and $\pi(2)<\pi(3)$, hence $\pi(1)+1<\pi(s-2)$ holds true if $s\geq 5$, which is the case by hypothesis.  

The height of $\Gamma$ is at most $n$, since every grid row intersecting $\Gamma$ contains a node of $P$ or intersects a subtree of $P$. Concerning the width, note that $\Gamma$ intersects no grid line $x=i$ with $i>0$. Consider the smallest $i$ such that the line $\ell$ with equation $x=i$ intersects $B(\Gamma)$. 

\begin{itemize}
\item Suppose that $\ell$ intersects a tree among $T_{1},\dots,T_{\pi(s)}$. Each of these trees lies either between the lines $x=-\omega$ and $x=-1$, or between the lines $x=-2\omega-2$ and $x=-\omega-3$; hence $i\geq -2\omega-2$ and the width of $\Gamma$ is at most $3+2\omega \leq 3 + 2f(A)$, where $\omega\leq f(A)$ holds true since every tree among $T_{1},\dots,T_{\pi(s)}$ has at most $A$ nodes and by the definition of the function $f(n)$. 
\item Next, suppose that $\ell$ intersects a tree among $T_{\pi(s)+1},\dots,T_{k-1}$. The drawing of each of these trees has the right side of its bounding box on the line $x=-2$; also, each of these trees has at most $A$ nodes, hence it has width at most $f(A)$. It follows that the width of $\Gamma$ is at most $2+f(A)$. 
\item Finally, suppose that $\ell$ intersects a tree among $L_1,\dots,L_p$. The drawing of each of these trees has the right side of its bounding box on the line $x=-2$; also, each of these trees has at most $n-A$ nodes, hence it has width at most $f(n-A)$. It follows that the width of $\Gamma$ is at most $2+f(n-A)$. 
\end{itemize}

This concludes the proof of the lemma.
\end{proof}

It remains to describe the strong flat algorithm for the case in which $s\geq 8$.

{\bf Strong flat algorithm with $\bf s\geq 8$.} The geometric construction for this case is {\em the same} as the one for the inductive case of the (non-strong) flat algorithm from Section~\ref{se:weak-star-shaped}, however the drawing algorithms which are recursively invoked by the two constructions differ; refer to Fig.~\ref{fig:strong-s6}(c). 

First, every subtree of the leftmost and rightmost paths of $T$ different from $T_{\pi(1)+1}$ is recursively drawn by means of the strong flat algorithm. Denote by $C$ and $D$ the left and right subtrees of $v_{\pi(1)+2}$. Bell-like star-shaped drawings of $C$ and $D$ are recursively constructed by means of the strong bell-like algorithm, however there is one difference in the recursive construction of these drawings. Note that the spine $P$ of $T$ ``enters'' exactly one between $C$ and $D$ (recall that $P$ contains the nodes $v_{\pi(1)},v_{\pi(1)+1},v_{\pi(1)+2},v_{\pi(1)+3}$, hence $v_{\pi(1)+3}$ is the root of $C$ or $D$); let $X$ be the one between $C$ and $D$ whose root is $v_{\pi(1)+3}$ and $Y$ be the one between $C$ and $D$ whose root is different from $v_{\pi(1)+3}$. Then the strong bell-like algorithm is applied recursively for $Y$, while $X$ is drawn by means of the construction of the strong bell-like algorithm with $s\geq 5$, by using the subpath of $P$ between $v_{\pi(1)+3}$ and $v_k$ as the spine for it (that is, the spine is not recomputed for $X$ according to Lemma~\ref{le:decomposition}, but the path $(v_{\pi(1)+3},v_{\pi(1)+4},\dots,v_k)$ is used as spine instead). Notice that, since $\pi(1)<\pi(2)<\pi(3)<\pi(4)$, we have that $\pi(1)+3\leq \pi(4)$, hence the spine $(v_{\pi(1)+3},v_{\pi(1)+4},\dots,v_k)$ contains at least $5$ switches. 

The remainder of the construction is the same as for the inductive case of the (non-strong) flat algorithm from Section~\ref{se:weak-star-shaped}. Indeed, the nodes of the leftmost and rightmost paths of $T$ are assigned $x$-coordinate equal to $0$; further, all the recursively drawn subtrees are embedded in the plane so that the left sides of their bounding boxes lie on the line $x=1$ (the drawing of $D$ is rotated by $180^{\circ}$ before embedding it). Node $v_{\pi(1)+2}$ is assigned $x$-coordinate equal to $1$ plus the maximum $x$-coordinate assigned to any other node in the drawing. Every node different from $v_1$ and $v_{\pi(1)}$ is assigned the same $y$-coordinate as its right or left child, depending on whether it belongs to the leftmost or rightmost path of $T$, respectively. Distinct subtrees are arranged vertically so that, from bottom to top, the nodes of the leftmost path of $T$ appear first -- in reverse order -- and then the nodes of the rightmost path of $T$ appear next -- in straight order. Depending on whether $v_2$ is the left child (see Fig.~\ref{fig:strong-s6}(c) top) or the right child (see Fig.~\ref{fig:strong-s6}(c) bottom) of $v_1$, we respectively have that: 

\begin{itemize}
\item The bottom side of the bounding box of $T_{\pi(1)}$ is one unit above the top side of the bounding box of $D$; the bottom side of the bounding box of $D$ is one unit above $v_{\pi(1)}$; $v_{\pi(1)}$ is one unit above $v_{\pi(1)+1}$ and $v_{\pi(1)+2}$, which have the same $y$-coordinate; $v_{\pi(1)+1}$ and $v_{\pi(1)+2}$ are one unit above the top side of the bounding box of $C$; and the bottom side of the bounding box of $C$ is one unit above the top side of the bounding box of the left child of $v_{\pi(1)+1}$ and of its right subtree.
\item The top side of the bounding box of $T_{\pi(1)}$ is one unit below the bottom side of the bounding box of $C$; the top side of the bounding box of $C$ is one unit below $v_{\pi(1)}$; $v_{\pi(1)}$ is one unit below $v_{\pi(1)+1}$ and $v_{\pi(1)+2}$, which have the same $y$-coordinate; $v_{\pi(1)+1}$ and $v_{\pi(1)+2}$ are one unit below the bottom side of the bounding box of $D$; and the top side of the bounding box of $D$ is one unit below the bottom side of the bounding box of the right child of $v_{\pi(1)+1}$ and of its left subtree.
\end{itemize}

We have the following.

\begin{lemma} \label{le:strong-flat-6}
Suppose that $s\geq 8$. Then the strong flat algorithm constructs a bell-like star-shaped drawing whose height is at most $n$ and whose width is at most $5 + \max\{2f(A),f(n-A)\}$. 
\end{lemma} 

\begin{proof}
It is readily seen that $\Gamma$ is star-shaped and flat, and that its height is at most $n$. The width of the drawing is given by $2$, corresponding to the grid column $x=0$ and to the grid column containing $v_{\pi(1)+2}$, plus the width of a recursively drawn subtree. The latter is the maximum between $f(A)$ (this is the maximum width of any tree different from $X$ that is recursively drawn) and $3+\max\{2f(A),f(n-A)\}$, which is the maximum width of the constructed drawing of $X$, as given by Lemma~\ref{le:strong-bell-like-6}. This concludes the proof of the lemma.
\end{proof}

We are now ready to state the main theorem of this section.

\begin{theorem} \label{th:lr-outerplanar}
Every $n$-vertex outerplanar graph admits an outerplanar straight-line drawing with area $O\left(n\cdot 2^{\sqrt{2 \log_2 n}} \sqrt{\log n}\right)$. 
\end{theorem} 

\begin{proof}
Let $G$ be an $n$-vertex outerplanar graph and let $T$ be its dual tree. We apply the strong flat algorithm to $T$ (with a parameter $A$ that will be specified shortly), thus obtaining a drawing $\Gamma$. Lemmata~\ref{le:strong-flat-05}--\ref{le:strong-flat-6} ensure that $\Gamma$ is a flat star-shaped drawing with height $O(n)$. Points $p^*_{u}$ and $p^*_{v}$ satisfying Property~4 of a star-shaped drawing can be chosen in $\Gamma$ (in fact in any flat star-shaped drawing) so that the width and the height only increase by a constant number of units. Due to this consideration and to Lemma~\ref{le:star-shaped-correspondence}, in order to conclude the proof of the theorem it only remains to argue that the width of $\Gamma$ is in $O\left(2^{\sqrt{2 \log_2 n}} \sqrt{\log n}\right)$. This proof follows almost verbatim a proof by Chan~\cite{c-anlabdbt-99}. Recall that we denote by $f(n)$ the maximum width of a drawing of an $n$-node ordered rooted binary tree constructed by means of the strong flat or bell-like algorithm.

By Lemmata~\ref{le:strong-flat-05}--\ref{le:strong-flat-6} we have that $f(n)\leq \max\{8+2f(A),8+f(n-A)\}$. Iterating over the second term with the same value of $A$ we get $f(n)\leq \max\{8+2f(A),8+f(n-A)\}\leq \max\{8+2f(A),16+f(n-2A)\}\leq \max\{8+2f(A),24+f(n-3A)\}\leq \dots\leq \max \{8+2f(A),8(\frac{n}{A}-1)+f(A)\}\leq 2f(A)+8\frac{n}{A}+8$. 

We now set $A= \frac{n}{2^{\sqrt{2 \log_2 n}}}$, which gives us the recurrence $$f(n)\leq 2f\left(\frac{n}{2^{\sqrt{2 \log_2 n}}}\right)+ 8\cdot 2^{\sqrt{2 \log_2 n}}+8.$$ 

We remark that the iteration with the same value of $A$ mentioned in the computation of the recursive formula corresponds to using $A= \frac{n}{2^{\sqrt{2 \log_2 n}}}$ whenever we need to recursively draw a tree that has more than $\frac{n}{2^{\sqrt{2 \log_2 n}}}$ nodes. Once the tree size drops to $\frac{n}{2^{\sqrt{2 \log_2 n}}}$ or less, the drawing algorithms are applied recursively by recomputing the parameter $A$ based on the actual number of nodes in the tree that has to be drawn. 

It remains to solve the recurrence equation, which is done again as by Chan~\cite{c-anlabdbt-99}. Namely, set $m=2^{\sqrt{2 \log_2 n}}$, which is equivalent to $n=2^{\frac{(\log_2 m)^2}{2}}$, and set $g(m)=f(n)$. Then 

\begin{eqnarray*}
	g(m)&\leq&2f\left(\frac{2^{\frac{(\log_2 m)^2}{2}}}{m}\right)+8m+8=2f\left(2^{\left(\frac{(\log_2 m)^2}{2}-\log_2 m\right)}\right)+8m+8\\
	&\leq& 2f\left(2^{\left(\frac{(\log_2 m -1)^2}{2}\right)}\right)+8m+8=2g\left(\frac{m}{2}\right)+8m+8.
\end{eqnarray*}

The inequality $g(m)\leq 2g\left(\frac{m}{2}\right)+8m+8$ trivially implies that $g(m)\in O(m\log m)$, and hence that $f(n)\in O(2^{\sqrt{2 \log_2 n}} \sqrt{\log n})$, which concludes the proof of the theorem. 
\end{proof}

We conclude the section by remarking that the function $2^{\sqrt{2 \log_2 n}} \sqrt{\log n}$ is asymptotically smaller than any polynomial function of $n$; that is, for any constant $\varepsilon>0$, it holds true that $2^{\sqrt{2 \log_2 n}} \sqrt{\log n}<n^{\varepsilon}$ for sufficiently large $n$.

\section{Conclusions} \label{se:conclusions}

In the first part of the paper we studied LR-drawings of ordered rooted binary trees. We proved that an LR-drawing with optimal width for an $n$-node ordered rooted binary tree can be constructed in $O(n^{1.48})$ time. It would be interesting to improve the running time to an almost-linear bound; this might however require new insights on the structure of LR-drawings. We also proved that there exist $n$-node ordered rooted binary trees requiring $\Omega(n^{0.418})$ width in any LR-drawing; this bound is close to the upper bound of $O(n^{0.48})$ due to Chan~\cite{c-anlabdbt-99}. It seems unlikely that Chan's bound is tight (he writes ``The exponent $p=0.48$ is certainly not the best possible'') and the experimental evaluation we conducted seems to confirm that; thus the quest for LR-drawings with $o(n^{0.48})$ width is a compelling research direction. 

In the second part of the paper we established a strong connection between LR-drawings of ordered rooted binary trees and outerplanar straight-line drawings of outerplanar graphs. Namely we proved that, if an ordered rooted binary tree $T$ has an LR-drawing with a certain width and area, then the outerplanar graph $G$ whose dual tree is $T$ has an outerplanar straight-line drawing with asymptotically the same width and area. We also proved that $n$-vertex outerplanar graphs admit outerplanar straight-line drawings in almost-linear area; our area upper bound is $O\left(n\cdot 2^{\sqrt{2 \log_2 n}} \sqrt{\log n}\right)$. We believe that an $O(n\log n)$ area bound cannot be achieved by only squeezing the drawing in one coordinate direction while keeping the size of the drawing linear in the other direction; hence, we find very interesting to understand whether every outerplanar graph admits an outerplanar straight-line drawing whose width and height are both sub-linear. We remark that a similar question has a negative answer for general {\em planar graphs}~\cite{Val81} and even for {\em series-parallel graphs}, that are graphs that exclude $K_{4}$ as a minor (and form hence a super-class of outerplanar graphs): There exist $n$-vertex series-parallel graphs that require $\Omega(n)$ size in one coordinate direction in any straight-line planar drawing~\cite{f-lbarspg-10}.

\bibliographystyle{abbrv}
\bibliography{bibliography}

\end{document}